\documentclass[a4paper,english,cleveref,autoref,thm-restate,nolineno]{socg-lipics-v2021}
\usepackage[T1]{fontenc}
\usepackage{slantsc}
\usepackage[utf8]{inputenc}
\usepackage{amsmath,amsthm,amssymb}
\usepackage{tikz}
\usetikzlibrary{shapes.geometric}
\usetikzlibrary{angles,quotes}
\usepackage{siunitx}
\usepackage{mathtools}
\usepackage[utf8]{inputenc}
\usepackage{xcolor}
\usepackage{graphicx}
\usepackage{xparse}
\usepackage{xargs}
\usepackage{xspace}
\usepackage{algorithm}
\usepackage{algpseudocode} 

\AtBeginDocument{\DeclareFontShape{T1}{lmr}{m}{scit}{<->ssub*lmr/m/scsl}{}}

\title{Towards Sub-Quadratic Diameter Computation in Geometric Intersection Graphs}

\titlerunning{Towards Sub-Quadratic Diameter in Intersection Graphs}

\author{Karl Bringmann}{Saarland University and Max Planck Institute for Informatics, Saarland Informatics Campus, Saarbrücken, Germany}{bringmann@cs.uni-saarland.de}{}{This work is part of the project TIPEA that has received funding from the European Research Council (ERC) under the European Unions Horizon 2020 research and innovation programme (grant agreement No. 850979).}
\author{Sándor Kisfaludi‑Bak}{Aalto University, Espoo, Finland}{sandor.kisfaludi-bak@aalto.fi}{}{Part of this research was conducted while the author was at the Max Planck Institute for Informatics, and part of it while he was at the Institute for Theoretical Studies, ETH Zürich.}
\author{Marvin Künnemann}{Institute for Theoretical Studies, ETH Zürich, Switzerland}{marvin.kuennemann@eth-its.ethz.ch}{}{Research supported by Dr. Max Rössler, by the Walter Haefner Foundation,
and by the ETH Zürich Foundation. Part of this research was conducted while the author was at the Max Planck Institute for Informatics.}
\author{André Nusser}{BARC, University of Copenhagen, Denmark}{anusser@mpi-inf.mpg.de}{https://orcid.org/0000-0002-6349-869X}{Part of this research was conducted while the author was at Saarbrücken Graduate School of Computer Science and Max Planck Institute for Informatics. The author is supported by the VILLUM Foundation grant 16582.}
\author{Zahra Parsaeian}{Max Planck Institute for Informatics, Saarland Informatics Campus, Saarbrücken, Germany}{zparsaei@mpi-inf.mpg.de}{}{}

\authorrunning{
K. Bringmann,
S. Kisfaludi‑Bak,
M. Künnemann,
A. Nusser,
and Z. Parsaeian} 

\Copyright{
Karl Bringmann,
Sándor Kisfaludi‑Bak,
Marvin Künnemann,
André Nusser,
and Zahra Parsaeian} 



\category{} 

\relatedversion{}

\hideLIPIcs

\EventEditors{Xavier Goaoc and Michael Kerber}
\EventNoEds{2}
\EventLongTitle{38th International Symposium on Computational Geometry (SoCG 2022)}
\EventShortTitle{SoCG 2022}
\EventAcronym{SoCG}
\EventYear{2022}
\EventDate{June 7--10, 2022}
\EventLocation{Berlin, Germany}
\EventLogo{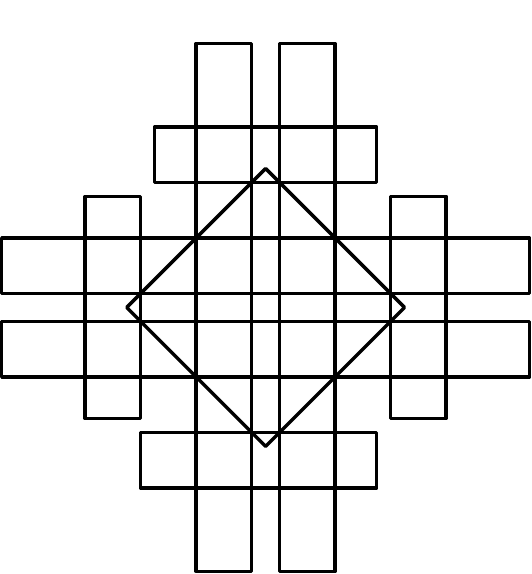}
\SeriesVolume{224}
\ArticleNo{XX}

\theoremstyle{plain}
\newtheorem{hypo}[theorem]{Hypothesis}

\newcommand{\diam}{\mathrm{diam}}
\newcommand{\dist}{\mathrm{dist}}
\newcommand{\grid}{\mathrm{GRID}}

\newcommand{\Oh}{\mathcal{O}}
\newcommand{\Otilde}{\widetilde{\Oh}}

\newcommand{\poly}{\mathrm{poly}}

\renewcommand{\leq}{\leqslant}
\renewcommand{\geq}{\geqslant}
\renewcommand{\rho}{\varrho}

\newcommand{\eps}{\varepsilon}

\newcommand{\Reals}{\mathbb{R}}

\newcommand{\etal}{\emph{et~al.}}

\begin{document}
	\maketitle
	
\begin{abstract}
	We initiate the study of diameter computation in geometric intersection
	graphs from the fine-grained complexity perspective. A geometric
	intersection graph is a graph whose vertices correspond to some shapes
	in $d$-dimensional Euclidean space, such as balls, segments, or
	hypercubes, and whose edges correspond to pairs of intersecting
	shapes. The diameter of a graph is the largest distance realized by a
	pair of vertices in the graph.
	
	Computing the diameter in near-quadratic time is possible in several
	classes of intersection graphs [Chan and Skrepetos 2019], but it
	is not at all clear if these algorithms are optimal, especially since
	in the related class of planar graphs the diameter can be computed in
	$\widetilde{\mathcal{O}}(n^{5/3})$ time [Cabello 2019, Gawrychowski et al. 2021].

	In this work we (conditionally) rule out sub-quadratic algorithms in
	several classes of intersection graphs, i.e., algorithms of running
	time $\mathcal{O}(n^{2-\delta})$ for some $\delta>0$. In particular, there are
	no sub-quadratic algorithms already for fat objects in small
	dimensions: unit balls in $\mathbb{R}^3$ or congruent equilateral
	triangles in $\mathbb{R}^2$. For unit segments and congruent equilateral
	triangles, we can even rule out strong sub-quadratic approximations
	already in $\mathbb{R}^2$. It seems that the hardness of approximation may
	also depend on dimensionality: for axis-parallel unit hypercubes
	in~$\mathbb{R}^{12}$, distinguishing between diameter 2 and 3 needs
	quadratic time (ruling out $(3/2-\varepsilon)$- approximations), whereas for
	axis-parallel unit squares, we give an algorithm that distinguishes
	between diameter $2$ and $3$ in near-linear time.

	Note that many of our lower bounds match the best known algorithms up to
	sub-polynomial factors. Ultimately, this fine-grained perspective may
	enable us to determine for which shapes we can have efficient
	algorithms and approximation schemes for diameter computation.
\end{abstract}
	
\section{Introduction}

The diameter of a simple graph $G=(V,E)$ is the largest distance realized
by a pair of its vertices; formally, it is $\diam(G)=\max_
{u,v\in V} \dist_G(u,v)$, where $\dist_G(u,v)$ is the number of edges
on a shortest path from $u$ to $v$. It is one of the crucial
parameters of a graph that can be computed in polynomial time.
Geometric intersection graphs are the standard model for wireless
communication networks~\cite{KuhnWZ08}, but more abstractly, they can
be used to represent networks where the connection of nodes relies on
proximity in some metric space. For a (slightly oversimplified)
example, consider a set of devices in the plane capable of receiving
and transmitting information in a range of radius $2$. These devices
form a communication network that is a unit disk graph. Indeed, two
devices can communicate with each other if and only if their distance
is at most $2$, i.e., if the unit disks centered at the devices have a
non-empty intersection. For our purposes, the underlying metric space
will be $d$-dimensional Euclidean space (henceforth denoted by
$\Reals^d$), and we will consider intersection graphs of common
objects such as balls and segments. For a set $F$ of objects  in
$\Reals^d$ (that is, $F\subset 2^{\Reals^d}$), the corresponding
intersection graph $G[F]$ has vertex set $F$ and edge set
$\{uv\mid u,v \in F,\, u\cap v \neq \emptyset\}$.

Computing the diameter in geometric intersection graphs is an important
task: if the graph represents a communication network, then the
diameter of the network can help estimate the time required to spread
information in the network, as the information needs to go through up
to $\diam(G)$ links to reach its destination. In large networks, it is
also indispensable to have near-linear time algorithms; it is
therefore natural to study if a given class of geometric intersection
graphs admits a near-linear time algorithm for exact or approximate
diameter computation.

The extensive literature on diameter computation serves as a good starting
point. The diameter of an $n$-vertex (unweighted) graph can be
computed in $\Oh(n^\omega \log n)$ expected time, where $\omega<2.37286$
is the exponent of matrix multiplication~\cite{Seidel95}. If the graph
has $m$ edges, then the diameter can also be computed in $\Oh
(mn)$ time~\cite{Thorup99}, which gives a near-quadratic running time
of $\Otilde(n^2)$ in case of sparse graphs, i.e., when $m=\Otilde
(n)$. In fact, these algorithms are capable of computing not only the
diameter, but also all pairwise distances in a graph, known as the
all pairs shortest paths problem.

On the negative side, we know that computing the diameter of a graph
cannot be done in $\Oh(n^{2-\eps})$ time under the Orthogonal Vectors
Hypothesis\footnote{See Section~\ref{sec:prelim} for the definitions
and some background on the hypotheses used in our lower bounds.}
(OV); in fact, deciding if the diameter of a sparse graph is at most $2$ or at least $3$
requires $n^{2-o(1)}$ time under OV~\cite{RodittyW13}\footnote{More precisely, Roditty and Vassilevska-Williams~\cite{RodittyW13} give a reduction from $k$-Dominating Set, which can be adapted to a reduction from OV as described in the beginning of Section~\ref{sec:lower}.}, which rules out
sub-quadratic $(3/2-\eps)$-approximations for all $\eps>0$.

In special graph classes however it is possible to compute the diameter in
sub-quadratic time. In planar graphs, an algorithm with running time
$\Oh(n^2)$ is very easy: one can just run $n$ breadth-first searches,
each of which take linear time because the number of edges is $\Oh
(n)$. It has been a long-standing open problem whether a truly
sub-quadratic algorithm exists for diameter computation, until the
breakthrough of Cabello~\cite{Cabello19}, who used Voronoi diagrams in
planar graphs. The technique was later improved by
Gawrychowski~\etal~\cite{GawrychowskiKMS21}, who obtained a running
time of $\Otilde(n^{5/3})$.

Certain geometric intersections graphs often behave similarly to planar
graphs. The most widely studied classes, (unit) disk and ball graphs
admit approximation schemes for maximum independent set, maximum
dominating set, and several other problems~\cite
{HochbaumM85,HuntMRRRS98,Chan03}, with techniques similar to planar
graphs. Unlike planar graphs, geometric intersection graphs can have
arbitrarily large cliques, but at least the maximum clique can be
approximated efficiently~\cite{BonamyBBCGKRST21}. In fact, planar
graphs are special disk intersection graphs by the circle packing
theorem~\cite{koebe1936kontaktprobleme}. When it comes to computing
the diameter, the similarity with planar graphs is not so easy to see.
Even getting near-quadratic diameter algorithms is non-trivial, as
geometric intersection graphs can be arbitrarily dense.

Chan and Skrepetos~\cite{Chan2019AllPairsSP} provide near-quadratic
($\Otilde(n^2)$) APSP algorithms for several graph classes, including
disks, axis-parallel segments, and fat triangles in the plane, and
cubes and boxes in constant-dimensional space.
Unit disk graphs have a ``weakly'' sub-quadratic algorithm (that is
poly-logarithmically faster than $\Oh(n^2)$)~\cite{ChanS16}. We are not
aware of any $\Oh(n^{2-\eps})$ algorithms for computing the diameter in intersection graphs of any planar shape.

\subparagraph*{Further related work.}
While computing the diameter is known to require time $n^{2\pm o(1)}$ already on sparse graphs (assuming the OV Hypothesis), an extensive line of research including~\cite{AingworthCIM99, RodittyW13, ChechikLRSTW14, CairoGR16, BackursRSWW21, DVVW19, Bonnet21stacs, DalirrooyfardW21, Li21, Bonnet21icalp, DLVW21arxiv} studies the (non-)existence of faster approximation algorithms. On the positive side, this includes in particular a folklore 2-approximation in time $\Otilde(m)$ and a $3/2$-approximation in time $\Otilde(m^{3/2})$~\cite{AingworthCIM99, RodittyW13, ChechikLRSTW14}, both already for weighted digraphs. Remarkably, these algorithms can be shown to be tight: \cite{RodittyW13, BackursRSWW21} establish that the $3/2$-approximation in time $\Otilde(m^{3/2})$ cannot be improved in either approximation guarantee or running time (assuming the $k$-OV Hypothesis), already for unweighted undirected graphs. The near-linear time $2$-approximation is conditionally optimal as well: For unweighted directed graphs, this has been proven independently in~\cite{DalirrooyfardW21, Li21}. For unweighted undirected graphs, following further work~\cite{Bonnet21icalp}, a resolution has been announced only very recently~\cite{DLVW21arxiv}. Thus, approximating the diameter in sparse graphs is quite well understood, including detailed insights into the full accuracy-time trade-off. In the context of our work, the challenge is to obtain a similar understanding for our setting of unweighted, undirected \emph{geometric} graphs, which are non-sparse in general.

Note that for graph classes that are non-sparse, a natural question is whether
diameter can be computed in $\Oh(m+n)$ time, i.e., linear time in the number of
edges plus vertices. The question has been studied by several authors: using
a variant of breadth-first search called \emph{lexicographic breadth-first
search}, one can find a vertex of very large eccentricity. In some classes,
we now know that there is an $\Oh(m+n)$ algorithm for diameter: notably, this holds in interval
graphs as well as $\{$claw,asteroidal triple$\}$-free graphs~\cite{DraganNB97,BrandstadtD03}. In many other graph classes (such as chordal graphs and asteroidal-triple-free graphs) we can get approximations for the diameter that differ only by a
small additive constant from the optimum~\cite
{DraganNB97,Dragan99,CorneilDHP01}. See~\cite{CorneilDK03} for an overview on
the connection of lexicographic BFS and diameter, and see~\cite
{Corneil04survey} for a survey on lexicographic BFS.

Another related direction is to consider edge weighted graph classes. In some
classes of geometric intersection graphs there is a natural weighting to
consider: for example in ball graphs, it is customary to draw the graph edges
with straight segments that connect the centers of the two adjacent disks.
The edges then have a natural weighting by their Euclidean length. This was
considered for unit disk graphs in the plane by Gao and Zhang~\cite
{GaoZ05}, who obtained a $(1+\eps)$-approximation for \textsc{Diameter} in
$\Oh(n^{3/2})$ time for any fixed $\eps>0$. A faster $
(1+\eps)$-approximation with running time $\Oh(n\log^2 n)$ for any fixed
$\eps>0$ was given by Chan and Skrepetos~\cite{ChanS19a}. Since the
underlying graph is not changed by this weighting, it is natural to think
that similar results should be possible also for unweighted unit disk graphs.
It remains an open question whether the complexity of diameter computation is
influenced by the presence of these Euclidean weights.

\subparagraph*{Our results.}

In this article, we show that most of the results of Chan and
Skrepetos~\cite{Chan2019AllPairsSP} cannot be significantly improved
under standard complexity-theoretic assumptions, even if we are only interested in the diameter instead of all pairs shortest paths. In particular, we
rule out sub-quadratic diameter algorithms for fat triangles and axis-aligned
segments in the plane, as well as for unit cubes in $\Reals^3$,
leaving only their $\Otilde(n^{7/3})$ algorithm for arbitrary segments
in $\Reals^2$ as well as their $\Otilde(n^2)$ algorithm for disks
without a matching lower bound.

The \textsc{Diameter} problem has as input a set of geometric objects in
$\Reals^d$ and a number $k$; the goal is to decide whether the diameter of
the intersection graph of the objects is at most $k$. The \textsc
{Diameter-$t$} problem is the same problem, but with $k$ set to the constant
number $t$.
We show the following lower bounds.
\begin{theorem}\label{thm:main}
For all $\delta>0$ there is no $\Oh(n^{2-\delta})$ time algorithm for 
\begin{itemize}
	\item \textsc{Diameter-3}
		in intersection graphs of unit segments in $\Reals^2$ under the
		OV Hypothesis. 
	\item \textsc{Diameter-3}
		in intersection graphs of congruent equilateral
		triangles in $\Reals^2$ under the OV Hypothesis.
	\item \textsc{Diameter} 
		in intersection graphs of unit balls in $\Reals^3$ under the
		OV Hypothesis.
	\item \textsc{Diameter} 
		in intersection graphs of axis-parallel unit cubes in $\Reals^3$ under the
		OV Hypothesis.
	\item \textsc{Diameter}
		in intersection graphs of axis-parallel line segments in $\Reals^2$
		under the OV Hypothesis.
	\item \textsc{Diameter-2}
		in intersection graphs of axis-parallel hypercubes in
		$\Reals^{12}$ under the Hyperclique Hypothesis.
\end{itemize}
\end{theorem}

Our results imply lower bounds for approximations. (See Section~\ref{sec:unitball} for a short proof.)

\begin{corollary}\label{cor:approxbounds}
Under the Orthogonal Vectors and Hyperclique Hypotheses, for all
$\delta,\eps>0$ there is no $\Oh(n^{2-\delta})$ time $
(4/3-\eps)$-approximation for \textsc{Diameter} in intersection graphs of
unit segments or congruent equilateral triangles in $\Reals^2$, and no $
(3/2-\eps)$-approximation in intersection graphs of axis-parallel hypercubes
in $\Reals^{\geq 12}$. Furthermore, for all $\delta>0$ there is no $
\Oh(n^{2-\delta}\poly(1/\eps))$ time approximation scheme that provides a $
(1+\eps)$-approximation for \textsc{Diameter} for any $\eps>0$ in
intersection graphs of axis-parallel unit segments in $\Reals^2$, or unit
balls or axis-parallel unit cubes in $\Reals^3$.
\end{corollary}

Theorem~\ref{thm:main} shows that sub-quadratic algorithms in many intersection graphs
classes are unlikely to exist; one must wonder if such algorithms are
possible at all? A notable case missing from our lower bounds are the
case of unit disks; indeed, it is possible that unit disk graphs enjoy
sub-quadratic diameter computation. More generally, it is an
interesting open question whether intersection graphs of so-called
pseudodisks admit sub-quadratic diameter algorithms. (Pseudodisks are
objects bounded by Jordan curves such that the boundaries of any pair of
objects have at most two intersection points.) We make a step
towards resolving this problem with the following theorem for
intersection graphs of axis-parallel unit squares --- since
axis-parallel unit squares are pseudodisks.

\begin{theorem}\label{thm:alg} 
	There is an $\Oh(n\log n)$ algorithm for
	\textsc{Diameter-2} in unit square graphs.
\end{theorem}

The algorithm is based on the insight that the problem can be simplified to the following: given \emph{skylines} $A,B$ and a list of axis-parallel squares $S$, check whether each pair $(a,b) \in A\times B$ is covered by some square $s\in S$. Since any axis-parallel square $s\in S$ covers \emph{intervals} in $A$ and $B$, this problem in turn reduces to checking whether the union of $|S|$ rectangles covers the $A\times B$ grid. Using near-linear skyline computation~\cite{KungLP75},  and a line sweep for the grid covering problem, we obtain a surprisingly simple $\Oh(n \log n)$ time algorithm (in contrast to the quadratic-time hardness in higher dimensions).

\subparagraph*{Organization.} After some preliminaries and the introduction of the complexity-theoretic hypotheses used in the paper, we present our algorithm for unit squares in Section~\ref{sec:alg}. Section~\ref{sec:lower} showcases our lower bound techniques. The lower bounds for unit segments, congruent equilateral triangles as well as for axis-parallel unit segments have a structure similar to two other lower bounds in Section~\ref{sec:lower}, and they can be found in the appendix.

\section{Preliminaries}\label{sec:prelim}

Let $G= (V, E)$ be a graph, and $u$ and $v$ be vertices in $G$. The
distance from $u$ to $v$ is denoted by  $\dist_G(u,v)$, and equals the
number of edges on the shortest path from $u$ to $v$ in $G$. The diameter of
$G$ is denoted by $\diam(G)$  and equals to $ \max_{u,v \in V} \dist_G
(u, v)$. The open and closed neighborhood of a vertex $v$ are $N(v)=
\{u\in V\mid uv \in E\}$ and $N[v]=\{v\}\cup N(v)$, respectively.
Let $A,B \subseteq V$ be sets of vertices. The diameter
of $A$ and $B$ is denoted by $\diam_G(A, B)
= \max_{(a,b) \in A \times B} \dist_{G}(a, b)$.
Finally, let $[n]$ denote the set $\{ 1, \dots , n \}$.

\subsection{Hardness assumptions}

We use two hypotheses from fine-grained complexity theory for our lower bounds. For an overview of this field, we refer to the survey~\cite{VassilevskaW18}.

\subparagraph*{Orthogonal Vectors Hypothesis.} 
Let \textsc{OV} denote the following problem: Given sets $A, B$ of $n$ vectors in $\{0,1\}^d$, determine whether there exists an \emph{orthogonal pair} $a\in A, b\in B$, i.e., for all $i\in [d]$ we have $(a)_i=0$ or $(b)_i = 0$.
Exhaustive search yields an $\Oh(n^2d)$ algorithm, which can be improved for small dimension $d=c\log n$ to $O(n^{2-1/O(\log(c))})$~\cite{AbboudWY15, ChanW16}. For larger dimensions $d=\omega(\log n)$, it is known~\cite{Williams05} that no $O(n^{2-\epsilon})$-time algorithm can exist unless the Strong Exponential Time Hypothesis~\cite{ImpagliazzoP01} fails. Thus, the Strong Exponential Time Hypothesis implies the following (so-called ``moderate-dimensional'') OV Hypothesis. 

\begin{hypo}[Orthogonal Vectors Hypothesis]
	For no $\epsilon > 0$, there is an algorithm that solves \textsc{OV} in time $\Oh(\mathrm{poly}(d) n^{2-\varepsilon})$.
\end{hypo}

By now, there is an extensive list of problems with tight lower bounds (including sub-quadratic equivalences) based on this assumption, see~\cite{VassilevskaW18}.

\subparagraph*{Hyperclique Hypothesis.} For $k\ge 4$, let \textsc
 {3-uniform $k$-Hyperclique} denote the following problem: Given a 3-uniform
 hypergraph $G=(V,E)$, determine whether there exists a \emph{hyperclique} of
 size $k$, i.e., a set $S\subseteq V$ such that for all $e\in \binom{S}{3}$,
 we have $e\in E$. By exhaustive search, we can solve this problem in time
 $\Oh(n^k)$ where $n=|V|$. Unlike the usual \textsc{$k$-Clique} problem
 in graphs, for which a $\Oh(n^{\omega k/3 + \Oh(1)})$ algorithm 
exists~\cite{NesetrilP85}, no techniques are known that would beat exhaustive search by a
 polynomial factor for the problem in hypergraphs. This has lead to the
 hypothesis that exhaustive search is essentially best possible.  

\begin{hypo}[Hyperclique Hypothesis]
	For no $\epsilon > 0$ and $k\ge 4$, there is an algorithm that would
	solve \textsc{3-uniform $k$-Hyperclique} in time $\Oh(n^{k-\epsilon})$.
\end{hypo}

See \cite{LincolnVWW18} for a detailed description of the plausibility
of this hypothesis. Tight conditional lower bounds (including fine-grained
equivalences) have been obtained, e.g., in~\cite{AbboudBDN18, BringmannFK19,
KunnemannM20, AnGIJKPN21}.

\section{Solving the Diameter-2 problem on unit square graphs}\label{sec:alg}

In this section, we are going to present an algorithm with running time
$\Oh(n\log n)$ for the \textsc{Diameter-$2$} problem for unit square
graphs.
For each unit square $v \in V$, we consider the center of $v$, denoted $\dot
{v}$, as the point representing $v$ in the plane; for a square set $X\subset
V$, we use $\dot X$ to denote the set of corresponding centers. Let $\dot
{G} = (\dot{V}, E)$ denote the graph on centers of squares in $G$. Hence, for
all $\{u, v\} \in E(G)$, there is an edge between $\dot{u}$ and $\dot
{v}$. Note that we will often use $\dot G$ and $G$ interchangeably.

Notice that a graph has diameter at most two if and only if for every pair of
vertices $u, v \in V:\; N[u]\cap N[v] \neq \emptyset$, i.e., there is a
square $w$ that both $u$ and $v$ have an intersection with or they intersect
each other. Equivalently, the square of side length $2$ centered at $\dot w$
must cover both $\dot u$ and $\dot v$. For a square $w$, let $w^2$ denote the
side-length-$2$ square of center $\dot w$. Thus, in order to decide whether
$\diam(G)\leq 2$, it is sufficient to check whether for every $u,v\in V$
there exists $w\in V$ such that $\dot u,\dot v \in w^2$.

For a set of points $P$ we define the \textit{top-left front}, $\textup{TLF}
(P)$, and \textit{bottom-right front}, $\textup{BRF}(P)$ as follows
(see Figure \ref{fig:front}a).
\begin{align*}
\textup{TLF}(P) &= \{p\in P \;|\; \forall q \in P\colon p_x \leq q_x \text{ or } p_y \geq q_y\}\\
\textup{BRF}(P) &= \{p\in P \;|\; \forall q \in P\colon p_x \geq q_x \text{ or } p_y \leq q_y\}
\end{align*}

Similarly, we define the \textit{top-right front}, $\textup{TRF}(P)$, and \textit
{bottom-left front}, $\textup{BLF}(P)$ as follows (see Figure \ref
{fig:front}b).
\begin{align*}
	\textup{TRF}(P) &= \{p \in P \,|\, \forall q \in P\colon p_x \geq q_x \text{ or } p_y \geq q_y\}\\
	\textup{BLF}(P) &= \{p \in P \,|\, \forall q \in P\colon p_x \leq q_x \text{ or } p_y \leq q_y\}
\end{align*}

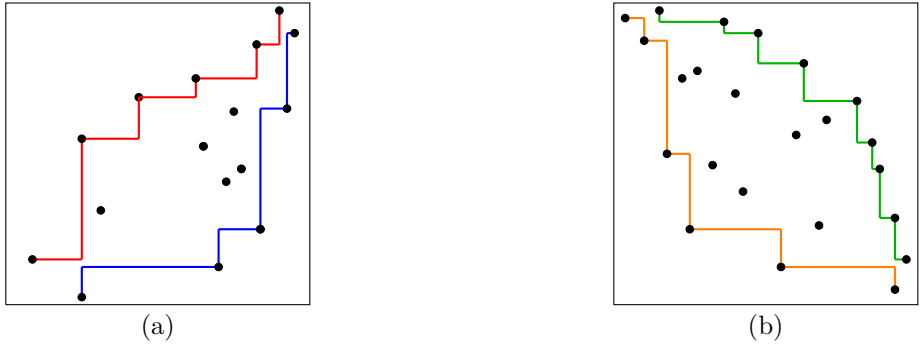
\begin{figure}[t] 
	\centering		
	\begin{tikzpicture}
		\draw node at (2,-.3) {(a)};
		\draw (0,0) -- (4,0) -- (4,4) -- (0,4) -- (0,0);
		
		\draw [color=blue, thick](3.8, 3.6) -- (3.7, 3.6);
		\draw [color=blue, thick](3.7, 2.6) -- (3.7, 3.6);
		\draw [color=blue, thick](3.7, 2.6) -- (3.35, 2.6);
		\draw [color=blue, thick](3.35, 2.6) -- (3.35, 1);
		\draw [color=blue, thick](2.8, 1) -- (3.35, 1);
		\draw [color=blue, thick](2.8, 1) -- (2.8, .5);
		\draw [color=blue, thick](2.8, .5) -- (1, .5);
		\draw [color=blue, thick](1, .5) -- (1, .1);
		\node[draw,circle,inner sep=1pt,fill] at (3.8,3.6) {};
		\node[draw,circle,inner sep=1pt,fill] at (2.9,1.63) {};
		\node[draw,circle,inner sep=1pt,fill] at (3.7,2.6) {};
		\node[draw,circle,inner sep=1pt,fill] at (2.6,2.1) {};
		\node[draw,circle,inner sep=1pt,fill] at (3.1,1.8) {};
		\node[draw,circle,inner sep=1pt,fill] at (3.35,1) {};
		\node[draw,circle,inner sep=1pt,fill] at (2.8,.5) {};
		\node[draw,circle,inner sep=1pt,fill] at (1,.1) {};

		\node[draw,circle,inner sep=1pt,fill] at (3,2.56) {};
\node[draw,circle,inner sep=1pt,fill] at (2.6,2.1) {};
		\node[draw,circle,inner sep=1pt,fill] at (3.1,1.8) {};
		\node[draw,circle,inner sep=1pt,fill] at (3.35,1) {};
		\node[draw,circle,inner sep=1pt,fill] at (1.75,2.75) {};
		\node[draw,circle,inner sep=1pt,fill] at (1.25,1.25) {};
		
		\draw [color=red, thick](3.6, 3.9) -- (3.6, 3.45);
		\draw [color=red, thick](3.6, 3.45) -- (3.3, 3.45);
		\draw [color=red, thick](3.3, 3.45) -- (3.3, 3);
		\draw [color=red, thick](3.3, 3) -- (2.5, 3);
		\draw [color=red, thick](2.5, 2.75) -- (2.5, 3);
		\draw [color=red, thick](2.5, 2.75) -- (1.75, 2.75);
		\draw [color=red, thick](1.75, 2.75) -- (1.75, 2.2);
		\draw [color=red, thick](1.75, 2.2) -- (1, 2.2);
		\draw [color=red, thick](1, 2.2) -- (1, .6);
		\draw [color=red, thick](1, .6) -- (.35, .6);
		\node[draw,circle,inner sep=1pt,fill] at (3.6,3.9) {};
		\node[draw,circle,inner sep=1pt,fill] at (3.3,3.45) {};
		\node[draw,circle,inner sep=1pt,fill] at (2.5,3) {};
		\node[draw,circle,inner sep=1pt,fill] at (1,2.2) {};
		\node[draw,circle,inner sep=1pt,fill] at (.35,.6) {};
		
\draw node at (10,-.3) {(b)};
		\draw (8,0) -- (12,0) -- (12,4) -- (8,4) -- (8,0);
		
		\draw [color=green!70!black!100!, thick](8.6, 3.9) -- (8.6, 3.75);
		\draw [color=green!70!black!100!, thick](9.45, 3.75) -- (8.6, 3.75);
		\draw [color=green!70!black!100!, thick](9.45, 3.75) -- (9.45, 3.6);
		\draw [color=green!70!black!100!, thick](9.45, 3.6) -- (9.9, 3.6);
		\draw [color=green!70!black!100!een, thick](9.9, 3.6) -- (9.9, 3.2);
		\draw [color=green!70!black!100!, thick](9.9, 3.2) -- (10.5, 3.2);
		\draw [color=green!70!black!100!, thick](10.5, 3.2) -- (10.5, 2.7);
		\draw [color=green!70!black!100!, thick](10.5, 2.7) -- (11.2, 2.7);
		\draw [color=green!70!black!100!, thick](11.2, 2.7) -- (11.2, 2.15);
		\draw [color=green!70!black!100!, thick](11.2, 2.15) -- (11.4, 2.15);
		\draw [color=green!70!black!100!, thick](11.4, 1.8) -- (11.4, 2.15);
		\draw [color=green!70!black!100!, thick](11.4, 1.8) -- (11.5, 1.8);
		\draw [color=green!70!black!100!, thick](11.5, 1.8) -- (11.5, 1.15);
		\draw [color=green!70!black!100!, thick](11.5, 1.15) -- (11.7, 1.15);
		\draw [color=green!70!black!100!, thick](11.7, .6) -- (11.7, 1.15);
		\draw [color=green!70!black!100!, thick](11.7, .6) -- (11.85, .6);
		\node[draw,circle,inner sep=1pt,fill] at (8.6,3.9) {};
		\node[draw,circle,inner sep=1pt,fill] at (9.9,3.6) {};
		\node[draw,circle,inner sep=1pt,fill] at (10.5,3.2) {};
		\node[draw,circle,inner sep=1pt,fill] at (11.2,2.7) {};
		\node[draw,circle,inner sep=1pt,fill] at (11.5,1.8) {};
		\node[draw,circle,inner sep=1pt,fill] at (11.85,.6) {};
		
		\node[draw,circle,inner sep=1pt,fill] at (8.9,3) {};
		\node[draw,circle,inner sep=1pt,fill] at (9.1,3.1) {};
		\node[draw,circle,inner sep=1pt,fill] at (9.3,1.85) {};
		\node[draw,circle,inner sep=1pt,fill] at (9.7,1.5) {};
		\node[draw,circle,inner sep=1pt,fill] at (9.6,2.8) {};
		\node[draw,circle,inner sep=1pt,fill] at (9.45,3.75) {};
		\node[draw,circle,inner sep=1pt,fill] at (10.4,2.25) {};
		\node[draw,circle,inner sep=1pt,fill] at (10.8,2.45) {};
		\node[draw,circle,inner sep=1pt,fill] at (11.4,2.15) {};
		\node[draw,circle,inner sep=1pt,fill] at (11.7,1.15) {};
		\node[draw,circle,inner sep=1pt,fill] at (10.7,1.05) {};
		
		\draw[color=orange, thick] (8.15, 3.8) -- (8.4, 3.8);
		\draw[color=orange, thick] (8.4, 3.5) -- (8.4, 3.8);
		\draw[color=orange, thick] (8.4, 3.5) -- (8.7, 3.5);
		\draw[color=orange, thick] (8.7, 3.5) -- (8.7, 2);
		\draw[color=orange, thick] (8.7, 2) -- (9, 2);
		\draw[color=orange, thick] (9, 2) -- (9, 1);
		\draw[color=orange, thick] (9, 1) -- (10.2, 1);
		\draw[color=orange, thick] (10.2, 1) -- (10.2, .5);
		\draw[color=orange, thick] (10.2, .5) -- (11.7, .5);
		\draw[color=orange, thick] (11.7, .5) -- (11.7, .2);
		\node[draw,circle,inner sep=1pt,fill] at (11.7,.2) {};
		\node[draw,circle,inner sep=1pt,fill] at (8.15,3.8) {};
		\node[draw,circle,inner sep=1pt,fill] at (8.4,3.5) {};
		\node[draw,circle,inner sep=1pt,fill] at (8.7,2) {};
		\node[draw,circle,inner sep=1pt,fill] at (9,1) {};
		\node[draw,circle,inner sep=1pt,fill] at (10.2,.5) {};
	\end{tikzpicture}
	\caption{The skylines (or fronts) of a point set $P$. In figure (a), the points on the blue curve are $\textup{BRF}(P)$ and on the red curve are $\textup{TLF}(P)$. In figure (b), the points on the green curve are $\textup{TRF}(P)$ and on the orange curve are $\textup{BLF}(P)$.}
	\label{fig:front}
\end{figure}
	
\begin{lemma}\label{lem:checkfronts}
	The graph $G$ has diameter at most $2$ if and only if
	\[\max\Big(\diam_{\dot G}(\textup{BLF}(\dot{V}),\textup{TRF}(\dot{V})),\diam_{\dot G}(\textup{TLF}(\dot{V}),\textup{BRF}(\dot{V}))\Big)\leq 2.\]
\end{lemma}
\begin{proof}
	If $G$ has diameter at most two, then clearly any pair of subsets of $\dot
	V$ have diameter at most $2$ in $\dot G$. For the other
	direction, consider any pair $a,b\in \dot V$, and assume
	that $a_x\leq b_x$ and $a_y\leq b_y$. We prove that $\dist_{\dot G }(a,b)\leq 2$.

	Select $ a'\in \textup{BLF}(\dot V)$ such that $ a'_x\leq  a_x$ and $ a'_y\leq  a_y$, see
	Figure~\ref{fig:frontsdominate}. Similarly, select $ b'\in \textup{TRF}(\dot V)$
	such that $ b_x\leq  b'_x$ and $ b_y\leq  b'_y$. Then we can observe that
	the minimum bounding box of $\{ a',  b'\}$ covers the minimum bounding box
	of $\{ a,  b\}$. Since $\dist_{\dot G}(a',b')\leq \diam_{\dot G}(\textup{BLF}(\dot
	{V}),\textup{TRF}(\dot{V}))\leq 2$, there exists a square $w\in V$ such that $w^2$
	covers $\{ a',  b'\}$. Consequently, $w^2$ also covers $\{ a,  b\}$, and
	thus $\dist_{\dot G}(a,b)\leq 2$.

	Finally, the case $a_x>b_x$ and $a_y>b_y$ is symmetric, and the cases
	$a_x\!>\!b_x, a_y\!\leq\!b_y$ and $a_x\!\leq\! b_x, a_y\!>\! b_y$ are analogous with
	$\textup{TLF}$ and $\textup{BRF}$ instead of $\textup{TRF}$ and $\textup{BLF}$.
\end{proof}
	
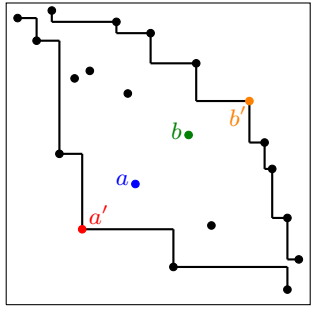
\begin{figure}[t] 
	\centering		
	\begin{tikzpicture}		
		\draw (8,0) -- (12,0) -- (12,4) -- (8,4) -- (8,0);
		
		\draw [thick](8.6, 3.9) -- (8.6, 3.75);
		\draw [thick](9.45, 3.75) -- (8.6, 3.75);
		\draw [thick](9.45, 3.75) -- (9.45, 3.6);
		\draw [thick](9.45, 3.6) -- (9.9, 3.6);
		\draw [thick](9.9, 3.6) -- (9.9, 3.2);
		\draw [thick](9.9, 3.2) -- (10.5, 3.2);
		\draw [thick](10.5, 3.2) -- (10.5, 2.7);
		\draw [thick](10.5, 2.7) -- (11.2, 2.7);
		\draw [thick](11.2, 2.7) -- (11.2, 2.15);
		\draw [thick](11.2, 2.15) -- (11.4, 2.15);
		\draw [ thick](11.4, 1.8) -- (11.4, 2.15);
		\draw [thick](11.4, 1.8) -- (11.5, 1.8);
		\draw [thick](11.5, 1.8) -- (11.5, 1.15);
		\draw [thick](11.5, 1.15) -- (11.7, 1.15);
		\draw [thick](11.7, .6) -- (11.7, 1.15);
		\draw [thick](11.7, .6) -- (11.85, .6);

		\node[draw,circle,inner sep=1pt,fill] at (8.6,3.9) {};
		\node[draw,circle,inner sep=1pt,fill] at (9.9,3.6) {};
		\node[draw,circle,inner sep=1pt,fill] at (10.5,3.2) {};
		\node[color=orange][draw,circle,inner sep=1pt,fill] at (11.2,2.7) {};
		\small\draw[color=orange] node at (11.05, 2.5){$ b'$};
		\node[draw,circle,inner sep=1pt,fill] at (11.5,1.8) {};
		\node[draw,circle,inner sep=1pt,fill] at (11.85,.6) {};
		
		\node[draw,circle,inner sep=1pt,fill] at (8.9,3) {};
		\node[draw,circle,inner sep=1pt,fill] at (9.1,3.1) {};
		
		\node[color=blue][draw,circle,inner sep=1pt,fill] at (9.7,1.6) {};
		\small\draw [color=blue]node at (9.53,1.65) {$ a$};
		\node[draw,circle,inner sep=1pt,fill] at (9.6,2.8) {};
		\node[draw,circle,inner sep=1pt,fill] at (9.45,3.75) {};
		\node[color=green!50!black!100!][draw,circle,inner sep=1pt,fill] at (10.4,2.25) {};
		\small\draw[color=green!50!black!100!] node at (10.24, 2.3){$ b$};
		\node[draw,circle,inner sep=1pt,fill] at (11.4,2.15) {};
		\node[draw,circle,inner sep=1pt,fill] at (11.7,1.15) {};
		\node[draw,circle,inner sep=1pt,fill] at (10.7,1.05) {};
		
		\draw[thick] (8.15, 3.8) -- (8.4, 3.8);
		\draw[thick] (8.4, 3.5) -- (8.4, 3.8);
		\draw[thick] (8.4, 3.5) -- (8.7, 3.5);
		\draw[thick] (8.7, 3.5) -- (8.7, 2);
		\draw[thick] (8.7, 2) -- (9, 2);
		\draw[thick] (9, 2) -- (9, 1);
		\draw[thick] (9, 1) -- (10.2, 1);
		\draw[thick] (10.2, 1) -- (10.2, .5);
		\draw[thick] (10.2, .5) -- (11.7, .5);
		\draw[thick] (11.7, .5) -- (11.7, .2);
		
		\node[draw,circle,inner sep=1pt,fill] at (11.7,.2) {};
		\node[draw,circle,inner sep=1pt,fill] at (8.15,3.8) {};
		\node[draw,circle,inner sep=1pt,fill] at (8.4,3.5) {};
		\node[draw,circle,inner sep=1pt,fill] at (8.7,2) {};
		\node[red][draw,circle,inner sep=1pt,fill] at (9,1) {};
		\small\draw [color=red]node at (9.22,1.2) {$ a'$};
		\node[draw,circle,inner sep=1pt,fill] at (10.2,.5) {};
	\end{tikzpicture}
	\caption{Any square covering $a'$ and $b'$ also covers $a$ and $b$.}\label{fig:frontsdominate}
\end{figure}

Using Lemma~\ref{lem:checkfronts}, we are able to prove Theorem~\ref
{thm:alg}. 

\begin{proof}[Proof of Theorem~\ref{thm:alg}.]
	We start our algorithm by computing $\textup{TLF}(\dot V),\textup{TRF}(\dot V),\textup{BLF}
	(\dot V)$, and $\textup{BRF}(\dot V)$ in $\Oh(n\log n)$ time~\cite{KungLP75}. Let $\dot{P}=\textup{BLF}(\dot V)$ and
	$\dot{Q}=\textup{TRF}(\dot V)$. By Lemma~\ref{lem:checkfronts}, it is sufficient to
	show that in $\Oh(n\log n)$ time we can decide whether $\diam_{\dot G}
	(\dot{P},\dot{Q})\leq 2$; using the same algorithm for $\textup{BRF}(\dot V)$ and $\textup{TLF}
	(\dot V)$ will then get the desired running time.

	In order to check whether $N[\dot{p}]\cap N
	[\dot{q}] \ne \emptyset$ for all $(\dot{p},\dot{q}) \in \dot
	{P} \times \dot{Q}$, we do the following:
		Consider $\dot{P} = \{\dot{p}_1,\dots,\dot{p}_{|\dot{P}|}\}$ and $\dot{Q} = \{\dot{q}_1,\dots,\dot{q}_
	{|\dot{Q}|}\}$ in $x$-order. Also, let $\grid = [|\dot{P}|]\times[|\dot{Q}|]$ be a
	grid where $\dot{p}_i$ corresponds to the $i$-th row and $\dot{q}_j$
	corresponds to the $j$-th column.
	
	For each square $v \in V$, recall that $v^2$ denotes the square with the
	same center but twice the side length. For each $v\in V$, define
	$I_v \subseteq \{1, 2, \dots, |\dot{P}|\}$ such that $i \in I_v$ iff
	$v^2$ contains $\dot{p}_i$. Similarly, $J_v \subseteq
	\{1, 2, \dots, |\dot{Q}|\}$ such that $j \in J_v$ iff $v^2$
	contains $\dot{q}_j$. Since $v^2$ is an axis-parallel square, it covers
	intervals from both $\dot{P}$ and $\dot{Q}$, thus $I_v$ and $J_v$ consist of
	consecutive integers. Therefore, we can think of the sets $I_v \times
	J_v$ as rectangles in $\grid$.
	
	\begin{claim}\label{cl:union}
		We have $N[\dot{p}] \cap N[\dot{q}] \ne \emptyset$ for all $(\dot{p},\dot{q}) \in \dot{P}
		\times \dot{Q}$ if and only if the union of $I_v \times J_v$ over all
		squares $v\in V$ covers $\grid$.
	\end{claim}
	\begin{claimproof}
		If the union of all rectangles covers the whole grid, then for any
		pair $(\dot{p}_i,\dot{q}_j) \in \dot{P} \times \dot{Q}$ of centers, there is a rectangle
		$I_v\times J_v$ that covers $(i,j)$.
		Therefore, $v^2$ covers both $\dot{p}_i$ and $\dot{q}_j$. Thus, $\dot v$ is a
		shared neighbor of $\dot{p}$ and $\dot{q}$.
		
		If $N[\dot{p}] \cap N[\dot{q}] \ne \emptyset$ for all $ (\dot{p},\dot{q}) \in \dot{P} \times \dot{Q}$,
		then for each pair $(\dot{p}_i,\dot{q}_j)$ there is at least one square $v_
		{ij}$ such that $v^2_{ij}$ contains both $\dot{p}_i$ and $\dot{q}_j$. Hence, $
		(i,j)\in I_{v_{ij}} \times J_{v_{ij}}$ for each $
		(i,j)\in \grid$. As a result, the union of $I_v \times J_v$ over
		all squares $v^2$ covers $\grid$.
	\end{claimproof}
	
	Note that the problem in Claim \ref{cl:union} corresponds to determining
	whether a union of rectangles covers the full grid. This problem can
	be solved in $\Oh(n\log n)$ time with a plane sweep~\cite
	{Bentley,DBLP:journals/jal/LeeuwenW81}. The
	time needed to construct the rectangles in $\grid$ is
	$\Oh(n\log n)$ as there are $O(n)$ rectangles. This concludes the proof of
	Theorem~\ref{thm:alg}. 
\end{proof}

\section{Lower bounds based on the Orthogonal Vectors Hypothesis}\label{sec:lower}
In this section, we prove lower bounds for finding the diameter in various intersection graphs.

For a comparison to similar results on sparse graphs, let us briefly describe the result ruling out a $(3/2-\epsilon)$-approximation in time $\Oh(n^{2-\delta})$, for any $\eps,\delta>0$, due to Roditty and Vassilevska-Williams~\cite{RodittyW13}. While it is originally stated as a reduction from $k$-Dominating Set, we adapt it to give a reduction from OV: Given sets $A,B\subseteq \{0,1\}^d$, introduce \emph{vector nodes} for each $a\in A$ and $b\in B$ as well as \emph{coordinate nodes} for $k\in[d]$. Without loss of generality (see Section~\ref{sec:simplesegments}), one may assume that all vectors $a\in A$ have $(a)_{d-1} = 1$ and all vectors $b\in B$ have $(b)_{d}=1$. We connect each vector node $v\in A\cup B$ to the coordinate node $k\in [d]$ iff $(v)_k = 1$, and make all coordinate nodes a clique by adding all possible edges between coordinate nodes. The important observation is that (1) a pair $a\in A, b\in B$ has distance at most 2 iff there is a $k\in[d]$ such that $(a)_k = (b)_k = 1$, i.e., $a,b$ do \emph{not} form an orthogonal pair, and (2) all other types of node pairs have distance at most 2. Thus, $A,B$ contains an orthogonal pair iff the diameter of the constructed graph is at least 3. Since the reduction produces a sparse graph with $O(n+d)$ nodes and $O(nd)$ edges in time $O(nd)$, any $O(m^{2-\delta})$-time algorithm distinguishing between diameter 2 and 3 would give a $O(n^{2-\delta} \mathrm{poly}(d))$-time OV algorithm, refuting the OV Hypothesis.  

Generally speaking, implementing this reduction using low-dimensional geometric graphs is problematic: we must be able to implement an arbitrary bipartite graph on a vertex set $L\times R$ where $|L|=n$ and $|R|=d$. Instead, in this section we implement two different types of reductions via geometric graphs; the main ideas are as follows:

\subparagraph*{Diameter-3 graphs (Sections~\ref{sec:simplesegments} and Appendix~\ref{sec:unitsegandtriangles}).} Instead of \emph{coordinate nodes}, we introduce \emph{1-entry nodes} $(v)_k$ for all $v\in A \cup B, k\in [d]$ with $(v)_k = 1$. This increases the number of nodes only to $O(nd)$, while allowing us to geometrically implement edges of the form $\{v, (v)_k\}$ for all $v\in A\cup B, k\in [d]$ with $(v)_k = 1$ and $\{(v)_k, (v')_k\}$ for all $v,v'\in A\cup B, k\in [d]$ with $(v)_k=(v')_k=1$. Now, a witness of non-orthogonality of $a,b$ is a 3-path $a - (a)_k - (b)_k - b$. By showing that all other distances are bounded by 3, we obtain hardness for the \textsc{Diameter-3} problem. See Sections~\ref{sec:simplesegments} and~\ref{sec:unitsegandtriangles} for details, including the use of an additional node to make all 1-entry nodes sufficiently close in distance.

\subparagraph*{(Non-sparse) Diameter-\boldmath$\Theta(d)$ graphs (Sections~\ref{sec:unitball} and Appendix~\ref{sec:unitseg}).} 
Instead of \emph{coordinate nodes} or \emph{1-entry nodes}, we introduce \emph{vector-coordinate nodes} $(v)_k$ for all $v\in A\cup B, k\in [d]$, irrespective of whether $(v)_k=1$. As opposed to previously, we do not create a constant diameter instance: The idea is to create an instance where the most distant pairs are of the form $(a)_1, (b)_d$ for $a\in A, b\in B$, and a non-orthogonality witness is a path of the form $(a)_1 \rightsquigarrow \dots \rightsquigarrow (a)_k \rightsquigarrow (b)_k \rightsquigarrow \dots \rightsquigarrow (b)_d$ with $(a)_k=(b)_k=1$. This construction requires us to implement perfect matchings between vector-coordinate gadgets $(a)_k$ for $a\in A$ and $(a')_{k+1}$ for $a'\in A$ if $a = a'$, as well as a gadget for implementing short connections for $(a)_k \rightsquigarrow (b)_k$ that check whether $(a)_k = (b)_k = 1$. Interestingly, this type of reduction generally produces dense graphs with $\Omega(n^2)$ edges, so this approach crucially exploits the expressive power of geometric graphs to give a subquadratic reduction. See Section~\ref{sec:unitball} and~\ref{sec:unitseg} for details, including a description of auxiliary nodes not mentioned here.

Finally, we remark that the reduction for unit hypercubes given in Section~\ref{sec:hyperclique-lb} has the most similar structure to the reduction by Roditty and Vassilevska-Williams~\cite{RodittyW13}, despite starting from a different hypothesis, and has similarities to~\cite[Theorem 14]{AnGIJKPN21}. We crucially exploit properties of the hyperclique problem to implement it using hypercube graphs.

\subsection{The Diameter-3 problem for line segment intersection graphs}\label{sec:simplesegments}
In this section, we are going to present a lower bound on the running time of the algorithm for the \textsc{Diameter-$3$} problem for line segment intersection graphs, such that vertices are line segments with any length, and there is an edge between a pair of line segments if they intersect. This serves as a warm-up for the slightly more complicated reductions below.

\begin{theorem}\label{line-seg}
	For all $\epsilon > 0$, there is no $\Oh(n^{2 - \epsilon})$ time algorithm for the \textsc{Diameter-$3$} problem for line segment intersection graphs, unless the OV Hypothesis fails.
\end{theorem}
Let  $A = \{a_1, a_2, \dots, a_n\}$ and $B = \{b_1, b_2, \dots, b_n\}$ be two sets of $n$ vectors in $\{0, 1\}^d$. We construct a set of segments such that the diameter of the  corresponding intersection graph is at most $3$ if and only if there is no orthogonal pair $(a, b) \in A \times B$.

Without loss of generality, we assume that for each $a_i \in A$ and $b_j \in B$, $\big((a_i)_{d -1}, (a_i)_{d}\big) = (1, 0)$ and $\big((b_j)_{d - 1}, (b_j)_{d}\big) = (0, 1)$, by adding two coordinates to the ends of the vectors. Note that adding these coordinates does not change whether vectors $a,b$ are orthogonal or not. 

For each vector $a_i \in A$, let $\bar{a}_i$ denote a zero-length line segment from $(i, 1)$ to $(i, 1)$. Analogously, for each vector $b_j \in B$, let $\bar{b}_j$ denote a line segment from $(j, -1)$ to $(j, -1)$. Furthermore, let $\ell$ be a line segment from $(1, 0)$ to $(d, 0)$, and let $\{w_1, w_2, \dots, w_{d}\}$ be $d$ different points on $\ell$ such that for all $k \in [d]$, $w_k$ is located at $(k, 0)$. Moreover, for each $a_i \in A$, if $(a_i)_k = 1$, we define a line segment $e_{i, k}$ from  $\bar{a}_i$ to $w_k$ (i.e., from $(i, 1)$ to $(k, 0)$). Analogously, for each $b_j \in B$, if $(b_j)_{k'} = 1$, we define a line segment $e'_{j, k'}$ from  $\bar{b}_j$ to $w_{k'}$ (i.e., from $(j, -1)$ to $(k', 0)$). Let $\bar{V}$ be the set of constructed line segments, and let $G$ be their intersection graph (see Figure \ref{line-seg-graphic}).

\begin{figure}[t]
	\centering
	\begin{tikzpicture}[scale=0.6,thick, main/.style = {draw,}]
		\draw [color = cyan](5.75,-4) -- (4.25,-2);
		\draw [color= cyan]node at (4.45,-3.25) {$e'_{j,k}$};
		\draw [color = cyan](5.25,0) -- (4.25,-2);
		\draw [color= cyan]node at (4.2,-.75) {$e_{i,k}$};
\node[draw,circle,inner sep=1pt,fill] at (.25,0) {};
		\draw node at (.25,0.4) {$\bar{a}_1$};
		\node[draw,circle,inner sep=1pt,fill] at (1.25,0) {};
		\draw node at (1.25,0.4) {$\bar{a}_2$};
		\node[draw,circle,inner sep=1pt,fill] at (2.25,0) {};
		\draw node at (2.25,0.4) {$\bar{a}_3$};
		\draw[dotted] (3.5,0) -- (3.75,0);
		\node[draw,circle,inner sep=1pt,fill] at (5.25,0) {};
		\draw node at (5.25,0.4) {$\bar{a}_i$};
		\draw[dotted] (6.5,0) -- (6.75,0);
		\node[draw,circle,inner sep=1pt,fill] at (8.25,0) {};
		\draw node at (8.25,0.4) {$\bar{a}_n$};
\draw node at (7.5,-2) {$\ell$};
		\draw (.25,-2) -- (6.5,-2);
		\node[draw,circle,inner sep=1pt] at (.25,-2) {};
		\draw node at (0,-1.6) {$w_1$};
		\node[draw,circle,inner sep=1pt] at (1.25,-2) {};
		\draw node at (1,-1.6) {$w_2$};
		\node[draw,circle,inner sep=1pt] at (2.25,-2) {};
		\draw node at (2,-1.6) {$w_3$};
		\node[draw,circle,inner sep=1pt] at (4.25,-2) {};
		\draw node at (4,-1.6) {$w_k$};
\node[draw,circle,inner sep=1pt] at (5.5,-2) {};
		\draw node at (5.35,-1.6) {$w_{d\!-\!1}$};
		\node[draw,circle,inner sep=1pt] at (6.5,-2) {};
		\draw node at (6.5,-1.6) {$w_{d}$};
\node[draw,circle,inner sep=1pt,fill] at (.25,-4) {};
		\draw node at (.25,-4.5) {$\bar{b}_1$};
		\node[draw,circle,inner sep=1pt,fill] at (1.25,-4) {};
		\draw node at (1.25,-4.5) {$\bar{b}_2$};
		\node[draw,circle,inner sep=1pt,fill] at (2.25,-4) {};
		\draw node at (2.25,-4.5) {$\bar{b}_3$};
		\draw[dotted] (3.75,-4) -- (4,-4);
		\node[draw,circle,inner sep=1pt,fill] at (5.75,-4) {};
		\draw node at (5.75,-4.5) {$\bar{b}_j$};
		\draw[dotted] (7,-4) -- (7.25,-4);
		\node[draw,circle,inner sep=1pt,fill] at (8.25,-4) {};
		\draw node at (8.25,-4.5) {$\bar{b}_n$};

	\end{tikzpicture}
	\caption{Reducing orthogonal vectors to \textsc{Diameter-3} in intersection graphs of line segments.}
	\label{line-seg-graphic}
\end{figure}
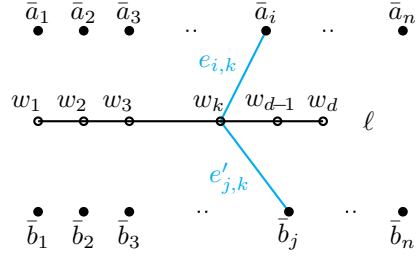

\begin{lemma}\label{lem:diamlessthan4}
	The sets $A$ and $B$ contain an orthogonal pair if and only if $\diam(G)\geq 4$.
\end{lemma}
Let $\bar{A}$ be the set of line segments corresponding to vectors in $A$. Analogously, let $\bar{B}$ be the set of line segments corresponding to vectors in $B$. To prove the lemma, we show that each pair of vertices is within distance at most $3$, unless it is in $\bar{A} \times \bar{B}$ (see Claim~\ref{cl:nofunnydistances} below and see Appendix~\ref{sec:claim} for its proof). The pairs in $\bar{A} \times \bar{B}$ have distance $4$ or $3$ depending on whether their corresponding vectors in $A\times B$ are orthogonal or not.

\begin{claim}\label{cl:nofunnydistances}
	$\dist(\bar{u}, \bar{v}) \leq 3$ for all $(\bar{u}, \bar{v})\in (\bar{V}\times \bar{V}) \setminus (\bar{A} \times \bar{B} \cup \bar{B} \times \bar{A})$.
\end{claim}
\begin{proof}[Proof of Lemma \ref{lem:diamlessthan4}]
If $a_i$ and $b_j$ are not orthogonal, then there is at least one $k \in [d]$ such that $(a_i)_k = (b_j)_k = 1$. Hence, the path $\bar{a}_i - e_{i,k} - e'_{j,k}-\bar{b}_j$ exists, and it has length 3.
If $a_i$ and $b_j$ are orthogonal, then there is no index $k $ such that $(a_i)_k = (b_j)_k = 1$. Consequently, there is no path of length 3 from $\bar{a}_i$ to $\bar{b}_j$, and $\dist(\bar{a}_i, \bar{b}_j) \geq 4$.
Together with Claim~\ref{cl:nofunnydistances} this proves the lemma.
\end{proof}

\begin{proof}[Proof of Theorem \ref{line-seg}]
	The above reduction creates a set of $N = \Oh(nd)$ segments in $\Oh(nd)$ time. If there is an algorithm solving \textsc{Diameter-3} in $\Oh(N^{2-\delta})$ time in segment intersection graphs, then combining this algorithm with the reduction would solve \textsc{OV} in time $\Oh(nd) + \Oh((nd)^{2 - \delta})=\Oh(n^{2-\delta}\mathrm{poly}(d))$, refuting the OV Hypothesis.
\end{proof}

\subsection{The diameter problem for unit ball graphs}\label{sec:unitball}

\begin{theorem}\label{thm:unit-ball}
	For all $\epsilon > 0$, there is no $\Oh(n^{2 - \epsilon})$ time algorithm
	for solving \textsc{Diameter} in unit ball graphs in $\Reals^3$ under the
	Orthogonal Vectors Hypothesis.
\end{theorem}

Let $A=\{a_1, a_2, \dots, a_n\}$ be a given set of vectors from $\{0,1\}^d$.
First, we construct graph $G(A)$ and show that $G(A)$ has diameter $\geq2d + 5$
if and only if there is an orthogonal pair of vectors in $A$. Next, we show
how $G(A)$ can be realized as an intersection graph of unit balls in
$\Reals^3$. Without loss of generality, assume that the all-one vector is
an element of $A$ (if it is not in $A$, then adding the all-one vector does
not change whether there is an orthogonal~pair.) 

We construct a graph $G(A)$ as follows. Let $C^T_1, \dots, C^T_
{2d}$ and $C^B_1, \dots, C^B_{2d}$ be cliques, such
that for all $k \in [2d]$, $C^T_k =\{v^T_{k,1}, \dots, v^T_{k, n}\}$, $C^B_k
= \{v^B_{k,1},\dots, v^B_{k, n}\}$, and $v^T_{k,i}$ and  $v^B_
{k,i}$ correspond to $a_i$ for all $i \in [n]$, see Figure~\ref
{fig:unitball}. We add a perfect matching between each pair $C^T_k$ and $C^T_
{k + 1}$ for all $k \in[2d - 1]$ such that there is an edge incident to $v^T_
{k,i}$ and $v^T_{k + 1, i}$ for all $i \in[n]$. Analogously, there is a
perfect matching between each pair $C^B_k$ and $C^B_{k + 1}$.

Let $M^T_1, \dots, M^T_{d}$ be cliques such that if $
(a_i)_k = 1$, then there is a vertex $m^T_{k,i}$ in $M^T_{k}$ that is
adjacent to $v_{k,i}^T$. Similarly, let $M^B_1, \dots, M^B_{d}$ be
cliques such that if $(a_i)_k = 1$, then there is a vertex $m^B_
{k,i}$ in $M^B_{k}$ that is adjacent to $v_{k,i}^B$. Notice that because of
the addition of the all ones vector, the cliques $M^T_k$ and $M^B_k$ are all
non-empty.

Finally, let $Q =\{q_1, q_2, \dots, q_d\}$ be a set of vertices such that
$q_k$ has edges to all vertices in $M^T_k$ and $M^B_k$ for all $k \in [d]$.

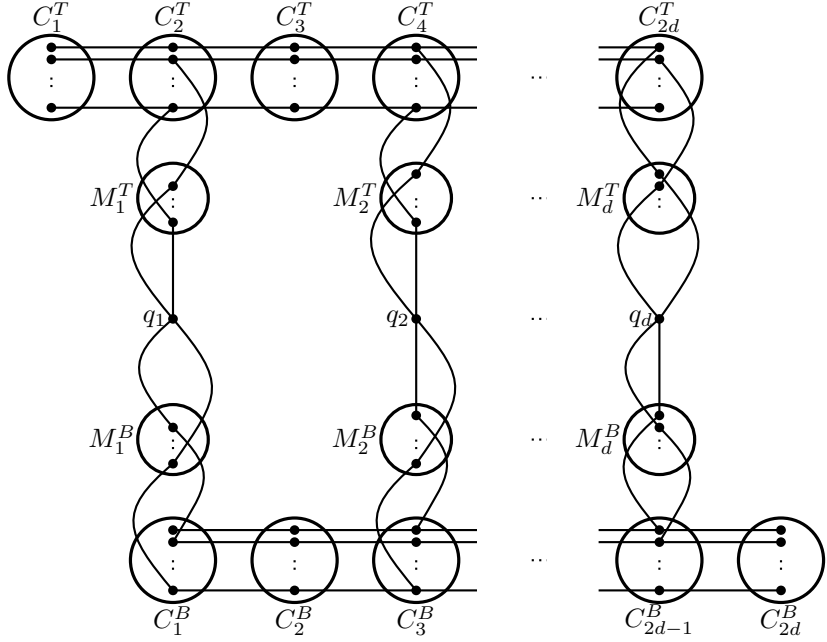
\begin{figure}[t]
	\centering
	\begin{tikzpicture}[thick, main/.style = {draw,circle},scale=0.8]\label{unit-ball-graphic}
		\draw node at (0,1) {$C^T_1$};
		\filldraw[color=black, fill=white, very thick](0,0) circle (.7);
		\node[draw,circle,inner sep=1pt,fill] at (0,.5) {};
		\node[draw,circle,inner sep=1pt,fill] at (0,.3) {};
		\draw[dotted] (0,0) -- (0,-.22);
		\node[draw,circle,inner sep=1pt,fill] at (0,-.5) {};
\draw node at (2,1) {$C^T_2$};
		\filldraw[color=black, fill=white, very thick](2,0) circle (.7);
		\node[draw,circle,inner sep=1pt,fill] at (2,.5) {};
		\node[draw,circle,inner sep=1pt,fill] at (2,.3) {};
		\draw[dotted] (2,0) -- (2,-.22);
		\node[draw,circle,inner sep=1pt,fill] at (2,-.5) {};
\draw node at (4,1) {$C^T_3$};
		\filldraw[color=black, fill=white, very thick](4,0) circle (.7);
		\node[draw,circle,inner sep=1pt,fill] at (4,.5) {};
		\node[draw,circle,inner sep=1pt,fill] at (4,.3) {};
		\draw[dotted] (4,0) -- (4,-.22);
		\node[draw,circle,inner sep=1pt,fill] at (4,-.5) {};
\draw node at (6,1) {$C^T_4$};
		\filldraw[color=black, fill=white, very thick](6,0) circle (.7);
		\node[draw,circle,inner sep=1pt,fill] at (6,.5) {};
		\node[draw,circle,inner sep=1pt,fill] at (6,.3) {};
		\draw[dotted] (6,0) -- (6,-.22);
		\node[draw,circle,inner sep=1pt,fill] at (6,-.5) {};
\draw[dotted] (7.9,0) -- (8.15,0);
\draw node at (10,1) {$C^T_{2d}$};
		\filldraw[color=black, fill=white, very thick](10,0) circle (.7);
		\node[draw,circle,inner sep=1pt,fill] at (10,.5) {};
		\node[draw,circle,inner sep=1pt,fill] at (10,.3) {};
		\draw[dotted] (10,0) -- (10,-.22);
		\node[draw,circle,inner sep=1pt,fill] at (10,-.5) {};
\draw node at (1,-2) {$M^T_{1}$};
		\filldraw[color=black, fill=white, very thick](2,-2) circle (.58);
\node[draw,circle,inner sep=1pt,fill] at (2,-1.8) {};
		\draw[dotted] (2,-2) -- (2,-2.22);
		\node[draw,circle,inner sep=1pt,fill] at (2,-2.4) {};
\draw node at (5,-2) {$M^T_{2}$};
		\filldraw[color=black, fill=white, very thick](6,-2) circle (.58);
		\node[draw,circle,inner sep=1pt,fill] at (6,-1.6) {};
\draw[dotted] (6,-2) -- (6,-2.22);
		\node[draw,circle,inner sep=1pt,fill] at (6,-2.4) {};
\draw[dotted] (7.9,-2) -- (8.15,-2);
\draw node at (9,-2) {$M^T_{d}$};
		\filldraw[color=black, fill=white, very thick](10,-2) circle (.58);
		\node[draw,circle,inner sep=1pt,fill] at (10,-1.6) {};
		\node[draw,circle,inner sep=1pt,fill] at (10,-1.8) {};
		\draw[dotted] (10,-2) -- (10,-2.22);
\draw node at (1.7,-4) {$q_1$};
		\node[draw,circle,inner sep=1pt,fill] at (2,-4) {};
		\draw node at (5.7,-4) {$q_2$};
		\node[draw,circle,inner sep=1pt,fill] at (6,-4) {};
		\draw[dotted] (7.9,-4) -- (8.15,-4);
		\draw node at (9.7,-4) {$q_d$};
		\node[draw,circle,inner sep=1pt,fill] at (10,-4) {};
\draw node at (1,-6) {$M^B_{1}$};
		\filldraw[color=black, fill=white, very thick](2,-6) circle (.58);
\node[draw,circle,inner sep=1pt,fill] at (2,-5.8) {};
		\draw[dotted] (2,-6) -- (2,-6.22);
		\node[draw,circle,inner sep=1pt,fill] at (2,-6.4) {};
\draw node at (5,-6) {$M^B_{2}$};
		\filldraw[color=black, fill=white, very thick](6,-6) circle (.58);
		\node[draw,circle,inner sep=1pt,fill] at (6,-5.6) {};
\draw[dotted] (6,-6) -- (6,-6.22);
		\node[draw,circle,inner sep=1pt,fill] at (6,-6.4) {};
\draw[dotted] (7.9,-6) -- (8.15,-6);
\draw node at (9,-6) {$M^B_{d}$};
		\filldraw[color=black, fill=white, very thick](10,-6) circle (.58);
		\node[draw,circle,inner sep=1pt,fill] at (10,-5.6) {};
		\node[draw,circle,inner sep=1pt,fill] at (10,-5.8) {};
		\draw[dotted] (10,-6) -- (10,-6.22);
\draw node at (2,-9) {$C^B_1$};
\filldraw[color=black, fill=white, very thick](2,-8) circle (.7);
		\node[draw,circle,inner sep=1pt,fill] at (2,-7.5) {};
		\node[draw,circle,inner sep=1pt,fill] at (2,-7.7) {};
		\draw[dotted] (2,-8) -- (2,-8.22);
		\node[draw,circle,inner sep=1pt,fill] at (2,-8.5) {};
\draw node at (4,-9) {$C^B_2$};
		\filldraw[color=black, fill=white, very thick](4,-8) circle (.7);
		\node[draw,circle,inner sep=1pt,fill] at (4,-7.5) {};
		\node[draw,circle,inner sep=1pt,fill] at (4,-7.7) {};
		\draw[dotted] (4,-8) -- (4,-8.22);
		\node[draw,circle,inner sep=1pt,fill] at (4,-8.5) {};
\draw node at (6,-9) {$C^B_3$};
		\filldraw[color=black, fill=white, very thick](6,-8) circle (.7);
		\node[draw,circle,inner sep=1pt,fill] at (6,-7.5) {};
		\node[draw,circle,inner sep=1pt,fill] at (6,-7.7) {};
		\draw[dotted] (6,-8) -- (6,-8.22);
		\node[draw,circle,inner sep=1pt,fill] at (6,-8.5) {};
\draw[dotted] (7.9,-8) -- (8.15,-8);
		\draw node at (10,-9) {$C^B_{2d - 1}$};
		\filldraw[color=black, fill=white, very thick](10,-8) circle (.7);
		\node[draw,circle,inner sep=1pt,fill] at (10,-7.5) {};
		\node[draw,circle,inner sep=1pt,fill] at (10,-7.7) {};
		\draw[dotted] (10,-8) -- (10,-8.22);
		\node[draw,circle,inner sep=1pt,fill] at (10,-8.5) {};
\draw node at (12,-9) {$C^B_{2d}$};
		\filldraw[color=black, fill=white, very thick](12,-8) circle (.7);
		\node[draw,circle,inner sep=1pt,fill] at (12,-7.5) {};
		\node[draw,circle,inner sep=1pt,fill] at (12,-7.7) {};
		\draw[dotted] (12,-8) -- (12,-8.22);
		\node[draw,circle,inner sep=1pt,fill] at (12,-8.5) {};
\draw (0,.5) -- (2,.5);
		\draw (0,.3) -- (2,.3);
		\draw (0,-.5) -- (2,-.5);
\draw (4,.5) -- (2,.5);
		\draw (4,.3) -- (2,.3);
		\draw (4,-.5) -- (2,-.5);
\draw (4,.5) -- (6,.5);
		\draw (4,.3) -- (6,.3);
		\draw (4,-.5) -- (6,-.5);
\draw (6,.5) -- (7,.5);
		\draw (6,.3) -- (7,.3);
		\draw (6,-.5) -- (7,-.5);
\draw (10,.5) -- (9,.5);
		\draw (10,.3) -- (9,.3);
		\draw (10,-.5) -- (9,-.5);
\draw (4,-7.5) -- (2,-7.5);
		\draw (4,-7.7) -- (2,-7.7);
		\draw (4,-8.5) -- (2,-8.5);
\draw (4,-7.5) -- (6,-7.5);
		\draw (4,-7.7) -- (6,-7.7);
		\draw (4,-8.5) -- (6,-8.5);
\draw (6,-7.5) -- (7,-7.5);
		\draw (6,-7.7) -- (7,-7.7);
		\draw (6,-8.5) -- (7,-8.5);
\draw (9,-7.5) -- (10,-7.5);
		\draw (9,-7.7) -- (10,-7.7);
		\draw (9,-8.5) -- (10,-8.5);
\draw (10,-7.5) -- (12,-7.5);
		\draw (10,-7.7) -- (12,-7.7);
		\draw (10,-8.5) -- (12,-8.5);

\draw (2,.3)  to [out=-50,in=55,looseness=1.5] (2, -1.8);
		\draw (2,-1.8)  to [out=220,in=130,looseness=1.5] (2, -4);
		\draw (2,-.5)  to [out=220,in=130,looseness=1.5] (2, -2.4);
		\draw (2,-2.4)  -- (2, -4);
		
		\draw (6,.5)  to [out=-50,in=55,looseness=1.5] (6, -1.6);
		\draw (6,-1.6)  to [out=220,in=130,looseness=1.5] (6, -4);
		\draw (6,-.5)  to [out=220,in=130,looseness=1.5] (6, -2.4);
		\draw (6,-2.4)  -- (6, -4);
		
		\draw (10,.5)  to [out=220,in=130,looseness=1.5] (10, -1.6);
		\draw (10,.3)  to [out=-50,in=55,looseness=1.5] (10, -1.8);
		\draw (10,-1.6)  to [out=-50,in=55,looseness=1.5] (10, -4);
		\draw (10,-1.8)  to [out=220,in=130,looseness=1.5] (10, -4);
\draw (2,-5.8)  to [out=130,in=220,looseness=1.5] (2, -4);
		\draw (2,-6.4)  to [out=220,in=130,looseness=1.5] (2, -8.5);
		\draw (2,-5.8)  to [out=-45,in=60,looseness=1.5] (2, -7.7);
		\draw (2,-6.4)  to [out=55,in=-50,looseness=1.5] (2, -4);
		
		\draw (6,-5.6)  -- (6, -4);
\draw (6,-5.6)  to [out=-45,in=60,looseness=1.5] (6, -7.5);
		\draw (6,-6.4)  to [out=220,in=130,looseness=1.5] (6, -8.5);
		\draw (6,-6.4)  to [out=55,in=-50,looseness=1.5] (6, -4);
		
		\draw (10,-5.6)  -- (10, -4);
		\draw (10,-5.8)  to [out=130,in=220,looseness=1.5] (10, -4);
		\draw (10,-5.6)  to [out=220,in=130,looseness=1.5] (10, -7.5);
		\draw (10,-5.8)  to [out=-45,in=60,looseness=1.5] (10, -7.7);
\end{tikzpicture}
	\caption{Schematic picture of the graph $G(A)$.}
	\label{fig:unitball}
\end{figure}
\begin{lemma}\label{lem:diam_2d+4_unit-ball}
	The graph $G(A)$ has diameter at most $2d + 4$ iff $A$ has no orthogonal pair.
\end{lemma}
\begin{proof}
Assume that there is an orthogonal pair $
	(a_i, a_j) \in A$ such that $i\not = j$. Hence, $\sum^{n}_{k = 1}(a_i)_k
	(a_j)_k = 0$, which means that there is no $k \in [n]$ such that
	$(a_i)_k = (a_j)_k = 1$. Consequently, for all $k \in [d]$, the distance
	from $v^T_{2k, i} \in C_{2k}^T$ to $v^B_{2k - 1, j} \in C^B_{2k - 1}$
	is at least 5. Therefore, $2d+4 < \dist(v^T_{1,i}, v^B_{2d, j}) \le
	\diam(G(A))$.
	
	Now suppose that $A$ has no orthogonal pair.  We want to prove that $\diam
	(G(A)) \leq 2d + 4$. Since $A$ has no orthogonal pair, for each pair $(a_i, a_j)$ there is
	at least one $k \in [n]$ such that  $(a_i)_k = (a_j)_k = 1$.
	Therefore, there are cliques $M^T_k$ and $M^B_k$ that
	have the vertices $m^T_{k,i}$ and $m^B_{k,j}$ respectively. Since all
	vertices in $M^T_k$ and $M^B_k$ have an edge to $q_k$, we can reach
	$q_k$ from $v^T_{1, i}$ by a path of length $2k - 1 + 2$.
	Simultaneously, we can reach $q_k$ from $v^B_{2d, j}$ by a path of
	length  $2d - 2k + 1 + 2$. In total, this gives a path of length $2d +
	4$ between $v^T_{1, i}$ and $v^B_{2d, j}$. Furthermore, it is easy
	check that the distance of any pair of vertices where at least one
	vertex is outside $C^T_1\cup C^B_{2d}$ is at most $2d+4$. As a result,
	$\diam(G(A)) \leq 2d + 4$.
\end{proof}
\begin{lemma}
	$G(A)$ can be realized as an intersection graph of unit balls in $\mathbb{R}^3$.
\end{lemma}
\begin{proof}
	For converting $G(A)$ into an intersection graph of unit balls, we should consider each vertex in $G(A)$ as the center of a unit diameter ball, and for those vertices that are adjacent, their corresponding unit balls should intersect. To this end, we choose the following coordinates for the centers of the unit balls in $\mathbb{R}^3$:
	\begin{itemize}
		\item For all $k \in [2d]$ and $ i \in [n]$, the center point of $v_{k,i}^T \in C^T_k$  is $(k, \frac{i}{n}, 0)$.
		\item For all $k \in [d]$ and $ i \in [n]$, if $m_{k,i}^T \in M^T_i$ exists, then its center point is $(2k, \frac{i}{n}, -1)$.
		\item For all $k \in [d]$, the center point of $q_{k} \in Q$  is $(2k, \frac{1}{2}, -1.6)$.
		\item For all $k \in [d]$ and $ i \in [n]$, if $m_{i,j}^B \in M^B_i$ exists, its  center point is $(2k, \frac{i}{n}, -2.2)$.
		\item For all $k \in [2d]$ and $ i \in [n]$, the center point of $v_{k,i}^B \in C^B_k$  is $(k, \frac{i}{n}, -3.2)$.
	\end{itemize}
	The distance between center points that correspond to adjacent vertices should be at most 1. 
	For each two vertices in the same clique in $C^T_k$, $C^B_k$, $M^T_k$, and $M^B_k$ their center points differ only in the $y$-coordinate. Since this difference is at most $1 - 1/n<1$, they form a clique. For each two adjacent vertices in two different cliques, their center points differ either only in the $x$-, or only in the $z$-coordinate, by exactly 1, hence, they intersect. For a vertex in $Q$ and $M^T_k$, if $m^T_{k, i}$ exists, the distance between $m^T_{k, i}$ and $q_k$ is 
	\begin{center}
		$\sqrt{(2k - 2k)^ 2 + (\frac{i}{n} - \frac{1}{2})^2 + (-1 - (- 1.6))^2} = \sqrt{(\frac{i}{n} - \frac{1}{2})^2 + (0.6)^2} \leq  \sqrt{(\frac{1}{2})^2 + (0.6)^2} < 1$
	\end{center}
	The same argument holds for adjacent vertices in $Q$ and $M^B$. One can easily check that the non-adjacent vertices have distance strictly greater than $1$.
\end{proof}
\begin{proof}[Proof of Theorem \ref{thm:unit-ball}]
The construction creates a set of $N = \Oh(nd)$ balls in $\Oh(nd)$ time. If
	there is an algorithm to solve \textsc{Diameter} in $\Oh(N^{2-\delta})$ time in
ball graphs, then we could combine this construction with the algorithm, and
solve the \textsc{Orthogonal Vectors} problem in $\Oh(nd) + \Oh((nd)^
	{2 - \delta})=O(n^{2-\delta}\mathrm{poly}(d))$ time. This contradicts the Orthogonal Vectors Hypothesis, and concludes
the theorem.
\end{proof}

A simple transformation of this construction shows that we can realize $G(A)$ also as an intersection graph of axis-parallel unit cubes.

\begin{corollary}
For all $\epsilon > 0$, there is no $\Oh(n^{2 - \epsilon})$ time algorithm
	for solving \textsc{Diameter} in intersection graphs of axis-parallel unit cubes in $\Reals^3$ under the
	Orthogonal Vectors Hypothesis.	
\end{corollary}	

\begin{proof}
Let $P$ denote the set of centers constructed for unit balls. We rotate $P$ by $\pi/4$ around the $y$ axis, and scale $P$ by a factor of $\sqrt{2}$. Let $P'$ be the resulting set of points. Note that in $P$, all inter-clique edges were realized by a horizontal or vertical point pair of distance exactly~$1$. In $P'$, the corresponding pairs are diagonal segments in some plane perpendicular to the $y$-axis, therefore the unit side-length cubes centered at the corresponding pair of points will have a touching edge. It is routine to check that the unit side-length cubes centered at $P'$ realize the intersection graph $G(A)$.
\end{proof}

\begin{proof}[Proof of Corollary~\ref{cor:approxbounds}]
The lower bounds regarding constant-approximations in sub-quadratic time are
immediate consequences of our lower bounds for \textsc
{Diameter-2} and \textsc{Diameter-3}. Notice that our proofs for unit balls
and axis-parallel unit cubes in $\Reals^3$, as well as axis-parallel unit
segments in $\Reals^2$ use a construction where the resulting intersection
	graph has diameter $d^*=\Theta(d)$. Under OV, there exists no $
(1+\eps)$-approximation for these problems that would run in $n^
	{2-\delta}\poly(1/\eps)$ time, as setting $\eps=1/d^*=\Theta(1/d)$ would enable
	us to decide OV in $n^{2-\delta}\poly(d)$ time.
\end{proof}

\section{The Diameter-2 problem for hypercube graphs: a hyperclique lower bound}
\label{sec:hyperclique-lb}

\begin{theorem}
	For all $\epsilon > 0$ there is no $O(n^{2-\epsilon})$ algorithm for \textsc{Diameter-2} in unit hypercube graphs in $\Reals^{12}$, unless the Hyperclique Hypothesis fails.
\end{theorem}
\begin{proof}
	Observe that under the Hyperclique Hypothesis, it requires time $n^{6-o(1)}$ to find a hyperclique of size 6 in a given $3$-uniform hypergraph $G=(V,E)$. In fact, using a standard color-coding argument, we can assume without loss of generality that $G$ is $6$-partite: We have $V=V_1\cup \cdots \cup V_6$ for disjoint sets $V_i$ of size $n$ each, and any 6-hyperclique must choose exactly one vertex from each $V_i$.
	By slight abuse of notation, we view each $V_i$ as a disjoint copy of~$[n]$, i.e., node $j\in [n]$ in $V_i$ is different from node $j$ in $V_{i'}$ with $i'\ne i$.
Furthermore, by complementing the edge set, we arrive at the equivalent task of determining whether $G$ has an \emph{independent set} of size 6, i.e., whether there are $(v_1,\dots,v_6)\in V_1\times \cdots \times V_6$ such that $\{v_i,v_j,v_k\}\notin E$ for all distinct $i,j,k\in [6]$.
	Finally, for technical reasons, we assume without loss of generality that for each $v_i \in V_i$ and distinct $j,k \in [6]\setminus \{i\}$, there are $v_j\in V_j, v_k \in V_k$ with $\{v_i,v_j,v_k\}\in E$: To this end, simply add, for every $\ell \in [6]$,  a dummy vertex $v'_\ell$ to $V_\ell$, and add, for every $i,j,k$ and $v_j \in V_j,v_k\in V_k$, the edge $\{v_i',v_j,v_k\}$ to $E$, i.e., each dummy vertex is connected to all other pairs of vertices (including other dummy vertices). Observe that this yields an equivalent instance, since no dummy vertex can be contained in an independent set. 

	The reduction is given by constructing a set of $O(n^3)$ unit hypercubes in $\mathbb{R}^{12}$, which we specify by their centers. These (hyper)cubes are of three types: \emph{left-half cubes} representing a choice of the vertices $(x_1,x_2,x_3)\in V_1\times V_2 \times V_3$, \emph{right-half cubes} representing a choice of the vertices $(y_1,y_2,y_3)\in V_4\times V_5 \times V_6$ and \emph{edge cubes} representing an edge $\{v_i,v_j,v_k\}\in E$. In particular, the choice of a vertex in $V_i$ will be encoded in the dimensions $2i - 1$ and $2i$.

	Specifically, for each $(x_1,x_2,x_3)\in V_1\times V_2 \times V_3$ such that $\{x_1,x_2,x_3\}\notin E$, we define the center of the left-half cube $X_{x_1,x_2,x_3}$ as
	\[	\left(\frac{x_1}{n+1}, 1 - \frac{x_1}{n+1}, \frac{x_2}{n+1}, 1 - \frac{x_2}{n+1}, \frac{x_3}{n+1}, 1 - \frac{x_3}{n+1}, 2, \dots, 2\right). \]
	Similarly, for each $(y_1,y_2,y_3)\in V_4\times V_5 \times V_6$ such that $\{y_1,y_2,y_3\}\notin E$, we define the center of the right-half cube $Y_{y_1,y_2,y_3}$ as
	\[\left(2, \dots, 2, \frac{y_1}{n+1}, 1 - \frac{y_1}{n+1}, \frac{y_2}{n+1}, 1 - \frac{y_2}{n+1}, \frac{y_3}{n+1}, 1 - \frac{y_3}{n+1}\right).\]

	Finally, for each edge $e=\{v_i,v_j,v_k\} \in E$ not already in $V_1 \times V_2 \times V_3 \cup V_4 \times V_5 \times V_6$, we define a corresponding edge cube $E_{v_i,v_j,v_k}$ with the following center point: We set the $2i - 1$-th coordinate to $1 + \frac{v_i}{n+1}$, the $2i$-th coordinate to $2 - \frac{v_i}{n+1}$, and similarly we set the coordinates $2j - 1,2j,2k - 1,2k$ to $1 + \frac{v_j}{n+1}, 2 - \frac{v_j}{n+1}, 1 + \frac{v_k}{n+1}, 2 - \frac{v_k}{n+1}$, respectively, and we set all remaining coordinates to 1. For example, if $i=1,j=2,k=4$, the center point of $E_{v_1,v_2,v_4}$ is
\[	\left(1 + \frac{v_1}{n+1}, 2 - \frac{v_1}{n+1}, 1 + \frac{v_2}{n+1}, 2 - \frac{v_2}{n+1}, 1, 1, 1 + \frac{v_4}{n+1}, 2 - \frac{v_4}{n+1}, 1, 1, 1, 1\right). \]

	Let $S$ denote the set of all unit cubes $X_{x_1,x_2,x_3},Y_{y_4,y_5,y_6},E_{v_i,v_j,v_k}$ constructed above and let $G_S$ denote the geometric intersection graph of the unit cubes. We prove that $\diam(G_S) \le 2$ if and only if there is no independent set $(v_1,\dots,v_6)\in V_1\times \cdots \times V_6$ in the 3-uniform hypergraph $G=(V_1\cup \cdots \cup V_6, E)$: 

	\begin{enumerate}
		\item \textbf{Intra-set distances:} We have that the left- and right-half cubes as well as the edge cubes form cliques, i.e., $\dist_{G_S}(X_{x_1,x_2,x_3}, X_{x'_1,x'_2,x'_3}) \le 1$, $\dist_{G_S}(Y_{y_1,y_2,y_3}, Y_{y'_1,y'_2,y'_3}) \le 1$ and $\dist_{G_S}(E_{v_1,v_2,v_3}, E_{v'_1,v'_2,v'_3}) \le 1$:
			Observe that the center of each $X_{x_1,x_2,x_3}$ is contained in $[0,1]^6\times \{2\}^6$ and thus in a hypercube of side length at most 1. Thus, all cubes $X_{x_1,x_2,x_3}$ intersect each other, proving $\dist_{G_S}(X_{x_1,x_2,x_3}, X_{x'_1,x'_2,x'_3}) \le 1$. The remaining claims follow analogously by observing that the centers of $Y_{y_1,y_2,y_3}$ and $E_{v_1,v_2,v_3}$ are contained in $\{2\}^6 \times [0,1]^6$ and $[1,2]^{12}$, respectively, and thus also in hypercubes of side length at most~1.
		\item \textbf{Equality checks:} Let $x_1\in V_1, x_2\in V_2, x_3\in V_3$ and $v_i \in V_i,v_j \in V_j, v_k\in V_k$. Then $\dist_{G_S}(X_{x_1,x_2,x_3}, E_{v_i,v_j,v_k}) = 1$ iff $v_\ell = x_\ell$ whenever $\ell \in \{1,2,3\}\cap \{i,j,k\}$: Consider $\ell \in \{1,2,3\}\cap \{i,j,k\}$. Then the dimensions $(2\ell - 1, 2\ell)$ of $X_{x_1,x_2,x_3}$ and $E_{v_i,v_j,v_k}$ are equal to $(\frac{x_\ell}{n+1},1-\frac{x_\ell}{n+1})$ and $(1+\frac{v_\ell}{n+1}, 2-\frac{v_\ell}{n+1})$, respectively. Note that $(1+\frac{v_\ell}{n+1}) - \frac{x_\ell}{n+1}\le 1$ and $(2-\frac{v_\ell}{n+1}) - (1-\frac{x_\ell}{n+1})\le 1$ hold simultaneously iff $x_\ell = v_\ell$. All other dimensions $\ell'\notin \{1,2,3\}\cap \{i,j,k\}$ are trivially within distance 1, since dimensions $(2\ell' - 1, 2\ell')$ of $X_{x_1,x_2,x_3}$ and $E_{v_i,v_j,v_k}$ are $(2,2)$ and in $[1,2]^2$, respectively (if $\ell'\notin \{1,2,3\}$), or in $[0,2]^2$ and $(1,1)$, respectively (if $\ell'\notin \{i,j,k\}$). The analogous claim holds for distances between $Y_{y_1,y_2,y_3}$ and $E_{v_i,v_j,v_k}$. 
		\item \textbf{Edge distances:} We have that $\dist_{G_S}(X_{x_1,x_2,x_3},E_{v_i,v_j,v_k})\le 2$: By our technical assumption, we have that there is an edge $\{x_1, v'_4,v'_5\}\in E$ for some vertices $v'_4\in V_4$ and $v'_5\in V_5$. Thus, by the previous properties, we obtain that 
			\[\dist_{G_S}(X_{x_1,x_2,x_3},E_{v_i,v_j,v_k})\le \dist_{G_S}(X_{x_1,x_2,x_3},E_{x_1,v'_4,v'_5}) + \dist_{G_S}(E_{x_1,v'_4,v'_5}, E_{v_i,v_j,v_k}) \le 2.\]
		\item \textbf{Distances of left- and right-half cubes:} Let $x_1\in V_1, x_2\in V_2, x_3\in V_3$ and $y_1\in V_4, y_2\in V_5, y_3\in V_6$ such that $\{x_1,x_2,x_3\},\{y_1,y_2,y_3\}\notin E$ (thus, the left-half/right-half cubes for $\{x_1,x_2,x_3\},\{y_1,y_2,y_3\}$ exist). Then we have that $\dist_{G_S}(X_{x_1,x_2,x_3},Y_{y_1,y_2,y_3})> 2$ iff $(x_1,x_2,x_3,y_1,y_2,y_3)$ is an independent set in $G$: If the tuple $(x_1,x_2,x_3,y_1,y_2,y_3)$ is not an independent set, then there must be an edge $\{x_i,y_j,y_k\}$ or $\{x_i,x_j,y_k\}$ with $i,j,k\in [3]$, since $\{x_1,x_2,x_3\}$ and $\{y_1,y_2,y_3\}$ are non-edges. Consider the first case, the other is symmetric. Then by the equality-check property, that $\dist_{G_S}(X_{x_1,x_2,x_3}, E_{x_i,x_j,y_k})=1$ and $\dist_{G_S}(E_{x_i,x_j,y_k},Y_{y_1,y_2,y_3})=1$, which yields $\dist_{G_S}(X_{x_1,x_2,x_3}, Y_{y_1,y_2,y_3}) \le 2$. It remains to consider the case that the tuple $(x_1,x_2,x_3,y_1,y_2,y_3)$ is an independent set. Since there cannot be any edge between a left-half cube $X_{x_1',x_2',x_3'}$  -- which is contained in $(0,1)^6\times \{2\}^6$ -- and a right-half cube $Y_{y_1',y_2',y_3'}$ -- which is contained in $\{2\}^6 \times (0,1)^6$ --, the only way to reach $Y_{y_1,y_2,y_3}$ from $X_{x_1,x_2,x_3}$ via a path of length 2 would have to use some edge cube $E_{v_i,v_j,v_k}$. However, by the equality-check property, a path $X_{x_1,x_2,x_3} - E_{v_i,v_j,v_k} - Y_{y_1,y_2,y_3}$ would imply that the vertices chosen by $(x_1,x_2,x_3,y_1,y_2,y_3)$ would agree with $v_i,v_j,v_k$ in the sets $V_i,V_j,V_k$. Thus, we would have found an edge $\{v_i,v_j,v_k\}$ among $(x_1,x_2,x_3,y_1,y_2,y_3)$, contradicting the assumption that it is an independent set.
	\end{enumerate}

	Finally, observe that given a 3-uniform hypergraph $G$, we can construct the corresponding cube set $S$, containing $O(n^3)$ nodes, in time $O(n^3)$. Thus, if we had an $O(N^{2-\epsilon})$-time algorithm for determining whether an $N$-vertex unit cube graph $G_S$ has a diameter of at most $2$, we could detect existence of an independent set (or equivalently, hyperclique) of size 6 in $G$ in time $O(n^{6-3\epsilon})$, which would refute the Hyperclique Hypothesis.
\end{proof}

\bibliography{udgdiameter}

\clearpage
\appendix
\begin{figure}[h!]
	\centering
	\begin{tikzpicture}[scale=0.6,thick, main/.style = {draw,}]
\draw (0,0) -- (0, 1.5);
		\draw[dotted] (0,2) -- (0, 2.4);
		\scriptsize
		\draw node at (0,2.7) {$\bar{a}_1$};
		
		\draw (.5,0) -- (.5, 1.5);
		\draw[dotted] (.5,2) -- (.5, 2.4);
		\scriptsize
		\draw node at (.5,2.7) {$\bar{a}_2$};
		
		\draw (1,0) -- (1, 1.5);
		\draw[dotted] (1,2) -- (1, 2.4);
		\scriptsize
		\draw node at (1,2.7) {$\bar{a}_3$};
		
		\draw (1.5,0) -- (1.5, 1.5);
		\draw[dotted] (1.5,2) -- (1.5, 2.4);
		\scriptsize
		\draw node at (1.5,2.7) {$\bar{a}_4$};
		
		\draw[dotted] (2,0) -- (2.4,0);
		\draw (3,0) -- (3, 1.5);
		\draw[dotted] (3,2) -- (3, 2.4);
		\scriptsize
		\draw node at (3,2.7) {$\bar{a}_n$};
		
		\draw[purple][to-to] (0, 3.2) -- (3,3.2);
		\draw [color=purple]node at (1.5,3.6) {$\frac{1}{n}$};
		
		\draw[color=purple][to-to] (3.3, 0) -- (3.3, 2.4);
		\draw [color=purple]node at (3.6,1.2) {$1$};

		\draw[color=purple][to-to] (-.3, -3.9) -- (-.3, 0);
		\small
		\draw [color=purple]node at (-2.2,-2) {$y^{\star} = \sqrt{1 - \frac{1}{n^2}}$};
		
		\draw [color=purple][to-to](0,-.2) -- (.5,-.2);
		\draw [color=purple]node at (0.25,-.8) {$\frac{1}{n^2}$};
		
		\node[draw,circle,inner sep=.01cm,fill] at (0,-4) {};
		\scriptsize
		\draw node at (-.1,-4.3) {$w_1$};
		\node[draw,circle,inner sep=.01cm,fill] at (.5,-4) {};
		\scriptsize
		\draw node at (.5,-4.3) {$w_2$};
		\node[draw,circle,inner sep=.01cm,fill] at (1,-4) {};
		\scriptsize
		\draw node at (1.1,-4.3) {$w_3$};
		\draw[dotted] (1.55,-4.3) -- (1.85, -4.3);
		\node[draw,circle,inner sep=.01cm,fill] at (2.5,-4) {};
		\scriptsize
		\draw node at (2.5,-4.3) {$w_{d + 2}$};
		\draw (-.5, -4) -- (3.5, -4);
		\draw[dotted] (3.95, -4) -- (4.35, -4);
		\draw node at (4.55,-4) {$\ell$};
		\draw[purple][to-to] (-.5, -4.7) -- (4.35,-4.7);
		\draw [color=purple]node at (1.85,-5) {$1$};

		\draw[color=green!60!black!100!] (1.5,0) -- (.5, -4);
		\draw [green!60!black!100!] (1.5, 0) -- +(75.6:.125cm);
		\draw [green!60!black!100!] (.5, -4) -- +(255.6:.125cm);
		\draw[color=green!60!black!100!] node at (1.5,-2) {$e_{4,2}$};

		\draw[color=pink] (.75,-.25) rectangle ++(1,.5);
		\draw[pink][dotted] (1.75, 0) -- (10.5, -1.3);
		\draw[color=pink] (10.5,-2.8) rectangle ++(4.4,3);s

		\draw (11,-1.5) -- (11, -.25);
		\draw (13,-1.5) -- (13, -.25);
		\draw [green!60!black!100!] (13, -1.5) -- +(75.6:1cm);
		\draw [purple] [to-to](13.2, -1.5) -- +(75.6:1cm);
		\draw [color=purple]node at (14.1, -1) {$< \frac{0.26}{n^2}$};
		\draw [green!60!black!100!] (13, -1.5) -- +(255.6:1cm);

	\end{tikzpicture}
	\caption{Construction with unit segments. The green edges $e_{i,k}$ can extend at most $\frac{0.26}{n^2}$ beyond the bottom endpoint of $\bar a_i$ and beyond $w_k$.}
	\label{fig:unit-segments}
\end{figure}
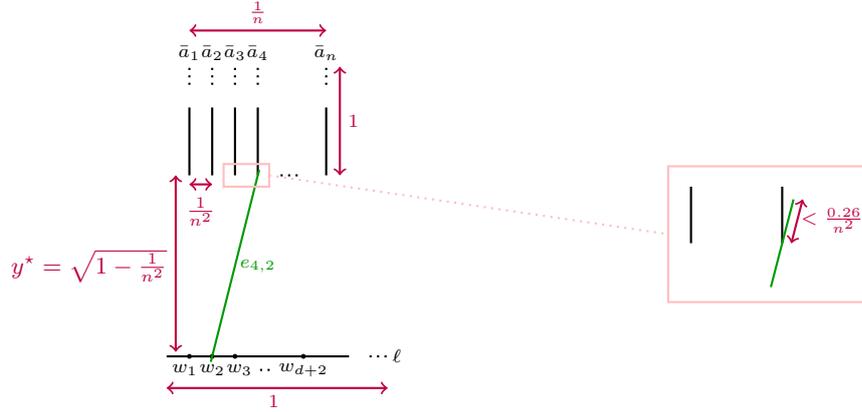
 
\section{The proof of Claim~\ref{cl:nofunnydistances}}\label{sec:claim}

\begin{claimproof} We need to consider the following distances:
	\begin{itemize}
		\item $\dist(\bar{a}_i, \bar{a}_j)$ for all $\bar{a}_i, \bar{a}_j \in \bar{A}$ and $\dist(\bar{b}_i, \bar{b}_j)$ for all $\bar{b}_i, \bar{b}_j \in \bar{B}$:\\
		By construction, $(a)_{d - 1} = 1$ for all $a \in A$, so we have a path $\bar{a}_i - e_{i, d - 1} - e_{j, d - 1} - \bar{a}_j$. As a result, $\dist(\bar{a}_i, \bar{a}_j)\leq 3$. Similarly, $(b)_d=1$ for all $b\in B$ implies that $\dist(\bar{b}_i, \bar{b}_j)\leq 3$.
		\item	$\dist(\ell, \bar{v}) $ for all $\bar{v}\in \bar{V}$:
		\begin{itemize}
			\item Every $e_{i, k}$ and $e'_{i,k} $ for all $i\in[n], k\in[d]$ is directly connected to $\ell$. As a result,  $\dist(e_{i, k}, \ell)=\dist(e'_{i,k}, \ell) = 1 $ for all $i\in[n], k\in[d]$.
			\item Every line segment in $(\bar{A} \cup \bar{B})$ is connected to $\ell$ via $e_{i, d-1}$ and $e'_{i, d}$, respectively. As a result,  $\dist(\ell, \bar{v}), \leq 2 $ for all $\bar{v} \in (\bar{A} \cup \bar{B})$.
		\end{itemize}
		In total, $\dist(\ell, \bar{v}) \leq 2$ for all $\bar{v} \in \bar{V}$.
		\item $\dist(e_{i, k}, \bar{v})$ and $\dist(e'_{j, k}, \bar{v})$ for all $\bar{v} \in \bar{V}, i,j \in [n], k \in [d]$:\\
		We have $\dist(e_{i, k}, \bar{v}) \leq \dist(e_{i, k}, \ell) + \dist(\ell, \bar{v}) \leq 1 + 2  = 3$.
		Similarly, $\dist(e'_{j, k}, \bar{v})\leq 3$.
	\end{itemize}
	Consequently, $\dist(\bar{u}, \bar{v}) \leq 3$ for all $(\bar{u}, \bar{v})\in (\bar{V}\times \bar{V}) \setminus (\bar{A} \times \bar{B} \cup \bar{B} \times \bar{A})$.
\end{claimproof}

\section{Unit segments and congruent equilateral triangles}\label{sec:unitsegandtriangles}
\begin{theorem}\label{thm:unit-line-seg}
	For all $\epsilon > 0$, there is no $\Oh(n^{2 - \epsilon})$ time algorithm
	for solving the \textsc{Diameter-3} problem for intersection graphs
	of unit line segments in $\Reals^2$, unless the Orthogonal Vectors
	Hypothesis fails.
\end{theorem}

\begin{proof}
We claim that the intersection graph of Subsection~\ref
{sec:simplesegments} can also be realized with unit segments, see Figure~\ref
{fig:unit-segments}. Indeed, set $w_k=(\frac{k}{n^2},0), k\in [d]$. Then
replace the segments $\bar a_i$ with the vertical unit segments $\{\frac{i}
{n^2}\}\times[y^*,y^*+1]$ so that consecutive segments have distance at most
$1/n^2$. Set $y^*=\sqrt{1-\frac{1}{n^2}}$, so that the lower endpoint of
$\bar a_i$ has distance at most $1$ and at least $y^*$ from any point $(\frac
{k}{n^2},0), k\in [d]$. If $n$ is large enough, then by the Taylor expansion
of $\sqrt{1-t}$ around $t=0$ we get that the distance between the lower
endpoint of $\bar a_i$ and $w_k$ is at least
\[y^*=\sqrt{1-\frac{1}{n^2}} =
1- \frac{1}{2n^2}-\Oh\left(\frac{1}{n^4}\right)>1-\frac{0.52}{n^2}.\]
Similarly, we set $b_j=\frac{j}{n^2}\times[-y^*-1,-y^*]$, so that the
construction has mirror symmetry on the $x$-axis. The segment $\ell$ is set
to $[0,1]\times 0$, and the segment $e_{i,k}$ is a segment containing the
bottom endpoint of $\bar a_i$ and $w_k$ so that the distance from $w_k$ to
the bottom endpoint of $e_{i,k}$ is equal to the distance from the bottom
endpoint of $\bar a_i$ to the top endpoint of $e_{ik}$. We define 
$e'_{j,k}$ in an analogous manner.

Notice that the distance of the bottom endpoint of $\bar a_i$ from the top
endpoint of $e_{i,k}$ is at most
\[\frac{1}{2}\Bigg(1-\min_{i\in [n],k\in [d]}
\dist\left(\left(\frac{i}{n^2},y^*\right),\left(\frac{k}{n^2},0
\right)\right)\Bigg)< \frac{1-y^*}{2}< \frac{0.26}{n^2}.\]
Since the pairwise distances between the segments $\bar a_k$ are at least
$1/n^2$, we have that $e_{i,k}$ intersects only $\bar a_k$ among these
segments. Similarly, for a fixed $k\in [d]$, the segments $e_{i,k}$ and $e'_
{j,k}$ extend beyond $w_k$ by at most $\frac{0.26}{n^2}$, so these extensions
are disjoint for distinct $k$ values as the distance between consecutive
$w_k$ points is $1/n^2$.

Since we have the same intersection graph as in Subsection~\ref
{sec:simplesegments}, the rest of the proof also works here.
\end{proof}

\begin{theorem}\label{thm:equilateral}
	For all $\epsilon > 0$, there is no $\Oh(n^{2 - \epsilon})$ time algorithm
	for solving the \textsc{Diameter-3} problem for intersection graphs
	of congruent equilateral triangles, unless the Orthogonal Vectors
	Hypothesis fails.
\end{theorem}

\begin{proof}
Consider each unit segment $s$ in the construction of Theorem~\ref
{thm:unit-line-seg} (except for $\ell$) replace $s$ with the equilateral
triangle that has $s$ as one of its sides, and points left(i.e., intersects
the $y$-axis). See Figure~\ref{fig:congruent-equilateras-triangles} for an
illustration. We replace $\ell$ with the equilateral triangle that has the
origin as its leftmost vertex, and it is symmetric on the $x$-axis. In what
follows, let us use the notations for the line segments also for the
corresponding triangles.

Notice that the triangles $\bar a_i$ form a clique, and also the triangles
$\bar b_j$ form a clique. This however does not affect the proof of
Lemma~\ref{lem:diamlessthan4}. It is easy to check that $\ell$ intersects
only the triangles $e_{i,k}$ and $e'_{j,k}$. It remains to show that $e_
{i,k}$ intersects only $\bar a_i$ among $\bar a_1,\dots \bar a_n$, and it
intersects exactly those $e'_{j,k'}$ where $k=k'$. The analogous statement
for all $e'_{j,k}$ can then be proven the same way. Proving this statement
ensures that the incidences between the sets $\{\bar a_1,\dots, \bar a_n\},
\{e_{1,1},\dots,e_{n,d}\},\{\ell\},\{e'_{1,1},\dots,e'_{n,d}\},
\{\bar b_1,\dots, \bar b_n\}$ are the same as in Subsection~\ref
{sec:simplesegments}, and thus the reduction works the same way.

Consider now a triangle $\Delta=e_{i,k}$. The intersection of $\Delta$ with
the half-plane $y\geq y^*$ is a triangle $\Delta^*$, where one of the
vertices is $v=(\frac{i}{n^2},y^*)$, the shortest side has length at most
$\sigma=\frac{0.26}{n^2}$, one of the angles is $\pi/3$, and the angle at $v$
is within the range $[\frac{\pi}{2}-0.01,\frac{\pi}{2}+0.01]$ (assuming that
$n$ is large enough). Consequently, the angle of $\Delta^*$ opposite the
shortest side is in the interval $[\frac{\pi}{6}-0.01,\frac{\pi}
{6}+0.01]$. By the law of sines, the longest side of $\Delta$ has length at
most $\frac{\sigma\cdot \sin(\pi/3)}{\sin(\frac{\pi}{6}-0.01)}<\frac{0.46}
{n^2}$. Consequently, $\Delta^*$ is contained in the rectangle $R=[\frac
{i-0.46}{n^2},\frac{i+0.46}{n^2}]\times[y^*,y^*+\frac{0.26}{n^2}]$. Notice
that $R$ is disjoint from both $\bar a_{i-1}$ and $\bar a_
{i+1}$. Therefore $\Delta^*$ (and therefore $\Delta$) cannot intersect any
triangle $\bar a_{i'}$ for $i'\neq i$.

By symmetry, the triangle $\Delta_*$ formed by the intersection of $\Delta$
and the half-plane $y\leq 0$ fits in the rectangle $R^*=[\frac{k-0.46}
{n^2},\frac{k+0.46}{n^2}]\times[-\frac{0.26}{n^2},0]$. Notice that $\Delta_*$
intersects $e'_{j,k'}$ if and only if $k=k'$. Analogously, for a $e'_
{j,k'}$, let $\Delta'$ be the intersection of $e'_{j,k'}$ with the half-plane
$y\geq 0$. Then $\Delta'$ intersects $e_{i,k}$ if and only if $k=k'$. Thus
$e_{i,k}\cap e'_{j,k'}\neq \emptyset$ if and only if $k=k'$, as required.

Thus the construction has the required adjacencies between the groups
$\{\bar a_1,\dots, \bar a_n\}$, $\{e_{1,1}$, $\dots,e_{n,d}\}$, $\{\ell\}$,
$\{e'_{1,1},\dots,e'_{n,d}\}$, $\{\bar b_1,\dots, \bar b_n\}$, 
so Lemma~\ref{lem:diamlessthan4} also applies in the intersection graph
defined by this construction, which concludes the proof.
\end{proof}

\begin{figure}[t]
	\centering
	\resizebox{\textwidth}{!}{
	\begin{tikzpicture}[scale=0.6,thick, main/.style = {draw,fill=purple!30,angle radius=3mm, 
			angle eccentricity=2.5, 
			left, inner sep=.5pt}]
		
		\fill[white!80!black] (12.15,-1.5) rectangle (19.75,0.5);
		\draw [color=black]node at (19,-1) {$R$};

		\draw [white] (16.35,.42) coordinate (a) -- (16, -1.5) coordinate (c) --
		(12.2,-1.5) coordinate  (b) -- cycle;
		\pic[color=purple][main, "$\SI{}{}$"] {angle = b--a--c};
		\pic[color=purple][main, "$\SI{}{}$"] {angle = a--c--b};
		\scriptsize
		\draw [color=purple]node at (15.75,-.2) {$\alpha$};
		\scriptsize
		\draw [color=purple,thick]node at (15.5,-.9) {$\beta$};

		\node[isosceles triangle,isosceles triangle apex angle=60, draw, rotate=180,minimum size =2cm] (T60)at (0,2){};
		\scriptsize
		\draw node at (1.1,4.2) {$\bar{a}_1$};
		\node[isosceles triangle,isosceles triangle apex angle=60, draw, rotate=180,minimum size =2cm] (T60)at (.5,2){};
		\scriptsize
		\draw node at (1.6,4.2) {$\bar{a}_2$};
		\node[isosceles triangle,isosceles triangle apex angle=60, draw, rotate=180,minimum size =2cm] (T60)at (1,2){};
		\scriptsize
		\draw node at (2.1,4.2) {$\bar{a}_3$};
		\node[isosceles triangle,isosceles triangle apex angle=60, draw, rotate=180,minimum size =2cm] (T60)at (1.5,2){};
		\scriptsize
		\draw node at (2.6,4.2) {$\bar{a}_4$};
		\node[isosceles triangle,isosceles triangle apex angle=60, draw, rotate=180,minimum size =2cm] (T60)at (2,2){};
		\scriptsize
		\draw node at (3.1,4.2) {$\bar{a}_5$};
		\node[isosceles triangle,isosceles triangle apex angle=60, draw, rotate=180,minimum size =2cm] (T60)at (3,2){};
		\scriptsize
		\draw node at (4.1,4.2) {$\bar{a}_n$};
\draw[purple][to-to] (1, 4.6) -- (4,4.6);
		\draw [color=purple]node at (2.5,5) {$\frac{1}{n}$};
		
		\draw[color=purple][to-to] (4.4, 0) -- (4.4, 4);
		\draw [color=purple]node at (4.7,2) {$1$};
		
		\node [color=green!60!black!100!][isosceles triangle,isosceles triangle apex angle=60, draw, rotate=176,minimum size =2cm] (T60)at (1.4,-1.75){};
		\draw[color=green!60!black!100!] node at (1.5,-2) {$e_{4,k}$};
		
		\draw[color=purple][to-to] (-1.1, -3.8) -- (-1.1, 0);
		\small
		\draw [color=purple]node at (-3,-2) {$y^{\star} = \sqrt{1 - \frac{1}{n^2}}$};
		
		\draw [color=purple][to-to](1,-.2) -- (1.5,-.2);
		\draw [color=purple]node at (.6,-.2) {$\frac{1}{n^2}$};
		
		\node[draw,circle,inner sep=.01cm,fill] at (2.32,-3.78) {};
		\draw node at (2.32,-4.08) {$w_k$};

		\draw  (8.4, -1.5) -- (8.4, 1.5);
		\draw  (8.4, -1.5) -- +(150:1.3cm);
		\draw (16, -1.5) -- (16, 1.5);
		\draw  (16, -1.5) -- +(150:5cm);
		\draw (23.6, -1.5) -- (23.6, 1.5);
		\draw  (23.6, -1.5) -- +(150:5cm);
		\draw  [color=green!60!black!100!] (16, -1.5) -- +(80:1.95cm);
		\draw  [color=green!60!black!100!] (16, -1.5) -- +(260:1cm);
		\draw  [color=green!60!black!100!] (16.35,.42) -- +(205:6cm);
\draw[purple][to-to] (12.2, -1.7) -- (16,-1.7);
		\draw [color=purple]node at (14.1,-2.2) {$< \frac{0.46}{n^2}$};
		
		\draw[purple][to-to] (16.55, .42) -- (16.2,-1.5);
		\draw [color=purple]node at (17.2,-.5) {$<\frac{0.26}{n^2}$};

		\draw[color=pink] (1.75,-.2) rectangle ++(1.6,.5);
		\draw[dotted] (16, -1.5) -- (8.9, -1.5);
		\draw[pink][dotted] (2.85, 0) -- (7, -.4);
		\draw[color=pink] (7,-2.8) rectangle ++(17.1,5.8);
\end{tikzpicture}
	}
	\caption{Construction with congruent equilateral triangles. The green triangle $e_{i,k}$ goes only slightly above the line $y=y^*$ and is unable to intersect any triangle $\bar a_j$ other than $\bar a_i$. The marked angles are $\alpha = \frac{\pi}{3}, \beta \in [\frac{\pi}{2} - 0.01, \frac{\pi}{2} + 0.01]$.}
	\label{fig:congruent-equilateras-triangles}
\end{figure}
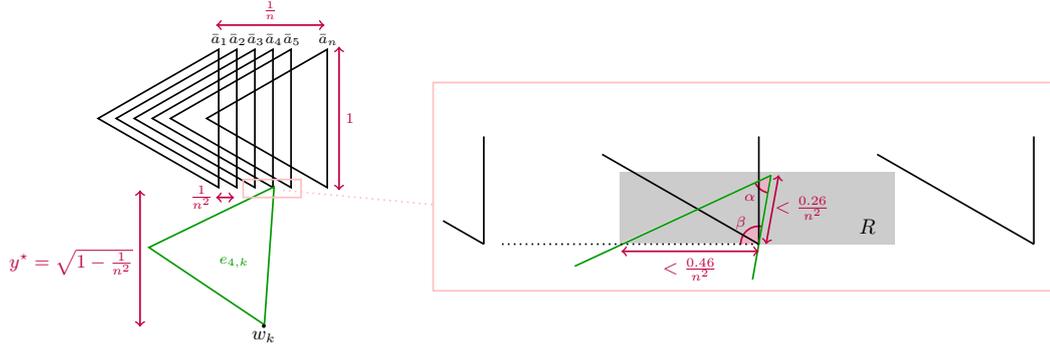

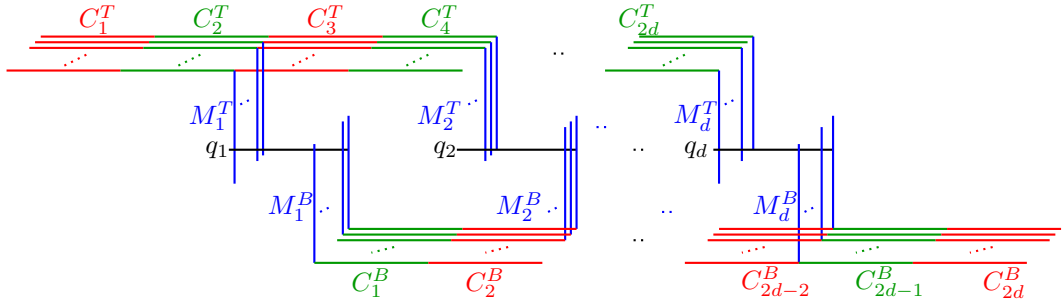
\begin{figure}[t]
	\centering
	\begin{tikzpicture}[scale=0.75, thick, main/.style = {draw,circle}]
		\draw [color=red]node at (1,.3) {$C^T_1$};
		\draw [color=red](0,0) -- (2,0);
		\draw [color=red](-.1,-.1) -- (1.9,-.1);
		\draw [color=red](-.2,-.2) -- (1.8,-.2);
		\draw[color=red][dotted] (.8,-.32) -- (.4,-.52);
		\draw [color=red](-.6,-.6) -- (1.4,-.6);

		\draw [color=green!60!black!100!]node at (3,.3) {$C^T_2$};
		\draw [color=green!60!black!100!](2,0) -- (4,0);
		\draw [color=green!60!black!100!](1.9,-.1) -- (3.9,-.1);
		\draw [color=green!60!black!100!](1.8,-.2) -- (3.8,-.2);
		\draw[color=green!60!black!100!][dotted] (2.8,-.32) -- (2.4,-.52);
		\draw [color=green!60!black!100!](1.4,-.6) -- (3.4,-.6);

		\draw [color=red]node at (5,.3) {$C^T_3$};
		\draw [color=red](4,0) -- (6,0);
		\draw [color=red](3.9,-.1) -- (5.9,-.1);
		\draw [color=red](3.8,-.2) -- (5.8,-.2);
		\draw[color=red][dotted] (4.8,-.32) -- (4.4,-.52);
		\draw [color=red](3.4,-.6) -- (5.4,-.6);

		\draw [color=green!60!black!100!]node at (7,.3) {$C^T_4$};
		\draw [color=green!60!black!100!](6,0) -- (8,0);
		\draw [color=green!60!black!100!](5.9,-.1) -- (7.9,-.1);
		\draw [color=green!60!black!100!](5.8,-.2) -- (7.8,-.2);
		\draw[color=green!60!black!100!][dotted] (6.8,-.32) -- (6.4,-.52);
		\draw [color=green!60!black!100!](5.4,-.6) -- (7.4,-.6);

\draw[dotted] (9,-.3) -- (9.25,-.3);

	\draw [color=green!60!black!100!]node at (10.5,.3) {$C^T_{2d}$};
	\draw [color=green!60!black!100!](10.5,0) -- (12.5,0);
	\draw [color=green!60!black!100!](10.4,-.1) -- (12.4,-.1);
	\draw [color=green!60!black!100!](10.3,-.2) -- (12.3,-.2);
	\draw[color=green!60!black!100!][dotted] (11.3,-.32) -- (10.9,-.52);
	\draw[color=green!60!black!100!](9.9,-.6) -- (11.9,-.6);

	\draw [color=blue] node at (3,-1.4) {$M^T_{1}$};
	\draw [color=blue] (3.4,-.6) -- (3.4,-2.6);
	\draw[color=blue] [dotted](3.5, -1.2) -- (3.7,-1.1);
	\draw [color=blue] (3.8,-.2) -- (3.8, -2.2);
	\draw [color=blue] (3.9,-.1) -- (3.9, -2.1);

	\draw [color=blue] node at (7,-1.4) {$M^T_{2}$};
\draw[color=blue] [dotted](7.5, -1.2) -- (7.7,-1.1);
	\draw [color=blue] (7.8,-.2) -- (7.8, -2.2);
	\draw [color=blue] (7.9,-.1) -- (7.9, -2.1);
	\draw [color=blue] (8,0) -- (8, -2);

\draw [color=blue] [dotted] (9.75,-1.6) -- (10,-1.6);

	\draw [color=blue] node at (11.5,-1.4) {$M^T_{d}$};
	\draw [color=blue] (11.9,-.6) -- (11.9,-2.6);
	\draw[color=blue] [dotted](12, -1.2) -- (12.2,-1.1);
	\draw [color=blue] (12.3,-.2) -- (12.3, -2.2);
\draw [color=blue] (12.5,0) -- (12.5, -2);
\draw node at (3.1,-2) {$q_1$};
	\draw (3.3, -2) -- (5.4, -2);
	\draw node at (7.1,-2) {$q_2$};
	\draw (7.3, -2) -- (9.4, -2);
	\draw [dotted](10.4, -2) -- (10.65, -2);
	\draw node at (11.5, -2) {$q_d$};
	\draw (11.8, -2) -- (13.9, -2);

	\draw [color=blue] node at (4.4,-3) {$M^B_{1}$};
	\draw [color=blue] (4.8,-1.9) -- (4.8,-4);
	\draw[color=blue] [dotted](4.9, -3.1) -- (5.1,-3);
\draw [color=blue] (5.3,-1.5) -- (5.3, -3.5);
	\draw [color=blue] (5.4,-1.4) -- (5.4, -3.4);

	\draw [color=blue] node at (8.4,-3) {$M^B_{2}$};
\draw[color=blue] [dotted](8.9, -3.1) -- (9.1,-3);
	\draw [color=blue] (9.2,-1.6) -- (9.2, -3.6);
	\draw [color=blue] (9.3,-1.5) -- (9.3, -3.5);
	\draw [color=blue] (9.4,-1.4) -- (9.4, -3.4);

		\draw[color=blue][dotted] (10.9,-3.1) -- (11.15,-3.1);

	\draw [color=blue] node at (12.9,-3) {$M^B_{d}$};
	\draw [color=blue] (13.3,-1.9) -- (13.3,-4);
	\draw[color=blue] [dotted](13.4, -3.1) -- (13.6,-3);
	\draw [color=blue] (13.7,-1.6) -- (13.7, -3.6);
\draw [color=blue] (13.9,-1.4) -- (13.9, -3.4);

	\draw [color=green!60!black!100!]node at (5.8,-4.35) {$C^B_1$};	
	\draw [color=green!60!black!100!](5.4,-3.4) -- (7.4,-3.4);
	\draw [color=green!60!black!100!](5.3,-3.5) -- (7.3,-3.5);
	\draw [color=green!60!black!100!](5.2,-3.6) -- (7.2,-3.6);
	\draw [color=green!60!black!100!][dotted] (6.2,-3.72) -- (5.8,-3.82);
	\draw [color=green!60!black!100!](4.8,-4) -- (6.8,-4);

	\draw [color=red]node at (7.8,-4.35) {$C^B_2$};	
	\draw [color=red](7.4,-3.4) -- (9.4,-3.4);
	\draw [color=red](7.3,-3.5) -- (9.3,-3.5);
	\draw [color=red](7.2,-3.6) -- (9.2,-3.6);
	\draw [color=red][dotted] (8.2,-3.72) -- (7.8,-3.82);
	\draw [color=red](6.8,-4) -- (8.8,-4);

	\draw [dotted] (10.4, -3.6) -- (10.65, -3.6);

	\draw [color=red]node at (12.9,-4.35) {$C^B_{2d - 2}$};	
	\draw [color=red](11.9,-3.4) -- (13.9,-3.4);
	\draw [color=red](11.8,-3.5) -- (13.8,-3.5);
	\draw [color=red](11.7,-3.6) -- (13.7,-3.6);
	\draw [color=red][dotted] (12.7,-3.72) -- (12.3,-3.82);
	\draw [color=red](11.3,-4) -- (13.3,-4);

	\draw [color=green!60!black!100!]node at (14.9,-4.35) {$C^B_{2d - 1}$};	
	\draw [color=green!60!black!100!](13.9,-3.4) -- (15.9,-3.4);
	\draw [color=green!60!black!100!](13.8,-3.5) -- (15.8,-3.5);
	\draw [color=green!60!black!100!](13.7,-3.6) -- (15.7,-3.6);
	\draw [color=green!60!black!100!][dotted] (14.7,-3.72) -- (14.3,-3.82);
	\draw [color=green!60!black!100!](13.3,-4) -- (15.3,-4);

	\draw [color=red]node at (16.9,-4.35) {$C^B_{2d}$};	
	\draw [color=red](15.9,-3.4) -- (17.9,-3.4);
	\draw [color=red](15.8,-3.5) -- (17.8,-3.5);
	\draw [color=red](15.7,-3.6) -- (17.7,-3.6);
	\draw [color=red][dotted] (16.7,-3.72) -- (16.3,-3.82);
	\draw [color=red](15.3,-4) -- (17.3,-4);

	\end{tikzpicture}
	\caption{Axis-parallel unit segments realizing the graph $G'(A)$.}
	\label{fig:axpar-unit-segments}
\end{figure}

\section{The diameter problem for axis-parallel unit segment graphs}\label{sec:unitseg}
In this section, we present a lower bound on the running time of the algorithm for the \textsc{Diameter} problem in intersection graphs of axis-parallel unit segments.

\begin{theorem}\label{thm:axis-par-unit-seg}
	For all $\epsilon > 0$, there is no $\mathcal{O}(n^{2 - \epsilon})$ time algorithm for solving \textsc{Diameter} in intersection graphs of axis-parallel unit segments, unless the Orthogonal Vectors Hypothesis fails.
\end{theorem}

\begin{proof}
Let $A = \{a_1, a_2, \dots, a_n\}$ be a set of $n$ vectors in $\{0, 1\}^d$.
First, we will construct a graph $G'(A)$ and show that $G'(A)$ has diameter
$2d + 4$ if and only if there is an orthogonal pair of vectors in $A$, and
then we will show how $G'(A)$ can be realized as an intersection graph of
axis parallel unit-segments in $\mathbb{R}^2$. Note that the constructed
graph $G'(A)$ will be very similar to the graph $G(A)$ constructed in
Section~\ref{sec:unitball}, with two crucial differences: (i) the cliques $M$
and $C$ will now be independent sets, and (ii) there will be additional edges
between $M^T_k$ and $C^T_{2k+1}$, as well as between $M^B_k$ and $C^B_
{2k}$. As in the other construction, assume without loss of generality that
the all-ones vector is in $A$.

We construct the graph $G'(A)$ as follows. Let $C^T_1, \dots, C^T_{2d}$ and
$C^B_1,\dots, C^B_{2d}$ be independent sets of size $n$, such that for all
$k \in [2d]$, $C^T_k =\{v^T_{k,1}, \dots, v^T_{k, n}\}$, $C^B_k = \{v^B_
{k,1}, \dots, v^B_{k, n}\}$. We add a perfect matching between each pair
$C^T_k$ and $C^T_{k + 1}$ for all $k \in [2d - 1]$ such that there is an edge
incident to $v^T_{k,i}$ and $v^T_{k + 1, i}$ for all $i \in[n]$. Analogously,
there is a perfect matching between each pair $C^B_k$ and $C^B_{k + 1}$ for
all $k \in [2d - 1]$. 

Let $M^T_1, \dots, M^T_{d}$ be independent sets of size at most $n$ such that
$M^T_i \subseteq\{m^T_{k,1}, \dots, m^T_{k, n}\}$, where $m^T_{k,i}\in M^T_k$
if and only if $(a_i)_k=1$. The vertex $m^T_{k,i}$ is adjacent to $v_
{2k,i}^T$; moreover, $m^T_{k,i}$ is also adjacent to $v_{2k+1,j}^T$ for all
$j \leq i$. We define the independent sets $M^B_1, \dots, M^B_{d}$ in
an analogous manner: $m^B_{k,i}$ is adjacent to $v_{2k-1,i}^B$, and also to
$v^B_{2k-2,j}$ for all $j \geq i$. Finally, let $Q = \{q_1, q_2, \dots,
q_d\}$ be a set of vertices such that $q_k$ is adjacent to all vertices of
$M^T_k$ and $M^B_k$.

\begin{lemma}
	The graph $G(A)$ has diameter $2d + 4$ if and only if there is no
	orthogonal pair in $A$.
\end{lemma}
The proof of this lemma is analogous to the proof of
Lemma~\ref{lem:diam_2d+4_unit-ball}.

\medskip
We now define a set of segments that realize $G'(A)$; see Figure~\ref
{fig:axpar-unit-segments}.
\begin{itemize}
		\item For all $k \in [2d]$ and $ i \in [n]$, the segment corresponding to $v^T_{k,i}\in C^T_k$ is horizontal, and its left endpoint is $(k-1+\frac{i}{4n}, \frac{i}{4n})$.
		\item For all $k \in [d]$ and $ i \in [n]$, the segment corresponding to $m_{k,i}^T \in M^T_i$ is vertical, and its top endpoint is $(2k+\frac{i}{4n}, \frac{i}{4n})$.
		\item For all $k \in [d]$, the left endpoint of the horizontal segment corresponding to $q_{k}$  is $(2k, -0.75)$.
		\item For all $k \in [d]$ and $ i \in [n]$, the segment corresponding to $m_{k,i}^B \in M^B_i$ is vertical, and its top endpoint is $(2k+0.75+\frac{i}{4n}, -0.75+\frac{i}{4n})$.
		\item For all $k \in [2d]$ and $ i \in [n]$, the segment corresponding to $v^B_{k,i}\in C^B_k$ is horizontal and its left endpoint is $(k+1.75+\frac{i}{4n}, -1.75+\frac{i}{4n})$.
\end{itemize}

It is routine to check that the intersection graph realized by these segments is $G'(A)$. The proof of Theorem~\ref{thm:axis-par-unit-seg} can now be wrapped up similarly to the proof of Theorem~\ref{thm:unit-ball}.
\end{proof}
  
\end{document}